\documentclass[12pt]{amsart}

\usepackage[T1]{fontenc}
\usepackage[utf8]{inputenc}
\usepackage[margin=2.5cm]{geometry}

\usepackage{amsmath,amssymb,amsthm}
\usepackage{cite}
\usepackage{enumerate}
\usepackage{xcolor}
\usepackage{hyperref}
\usepackage{cleveref}
\usepackage{thm-restate}
\usepackage{stmaryrd}

\hypersetup{
  pdftitle = {Profile and neighbourhood complexity of graphs excluding a minor and tree-structured graphs},
  pdfauthor = {Beaudou, Bok, Foucaud, Quiroz, Raymond},
  colorlinks = true,
  linkcolor = black!30!red,
  citecolor = black!30!green
}

\newtheorem{theorem}{Theorem}
\newtheorem{corollary}[theorem]{Corollary}
\newtheorem{definition}[theorem]{Definition}
\newtheorem{problem}[theorem]{Problem}
\newtheorem{proposition}[theorem]{Proposition}
\newtheorem{lemma}[theorem]{Lemma}

\def\cqedsymbol{\ifmmode$\lrcorner$\else{\unskip\nobreak\hfil
\penalty50\hskip1em\null\nobreak\hfil$\lrcorner$
\parfillskip=0pt\finalhyphendemerits=0\endgraf}\fi}

\newcommand{\N}{\mathbb{N}}
\newcommand{\T}{\mathcal{T}}
\newcommand{\intv}[2]{\left \{ #1, \dots, #2 \right \}}

\DeclareMathOperator{\nc}{nc}
\DeclareMathOperator{\pc}{pc}
\DeclareMathOperator{\diam}{diam}
\DeclareMathOperator{\md}{md}
\DeclareMathOperator{\ld}{ld}
\DeclareMathOperator{\SReach}{SReach}
\DeclareMathOperator{\WReach}{WReach}
\DeclareMathOperator{\scol}{scol}
\DeclareMathOperator{\wcol}{wcol}
\DeclareMathOperator{\dist}{dist}

\DeclareMathOperator{\cp}{Cap}
\DeclareMathOperator{\tw}{tw}
\DeclareMathOperator{\stw}{stw}
\DeclareMathOperator{\tl}{tl}

\title[Profile complexity of tree-structured graphs]{Profile and neighbourhood complexity of graphs excluding a minor and tree-structured graphs}

\author[L. Beaudou]{Laurent Beaudou}
\address[L. Beaudou]{Université Clermont Auvergne, CNRS, Clermont Auvergne INP,
Mines Saint-Étienne, LIMOS, 63000 Clermont-Ferrand, France}
\email{laurent.beaudou@uca.fr}
\author[J. Bok]{Jan Bok}
\address[J. Bok]{Université Clermont Auvergne, CNRS, Clermont Auvergne INP,
Mines Saint-Étienne, LIMOS, 63000 Clermont-Ferrand, France and Department of Algebra, Faculty of Mathematics and Physics, Charles University, Sokolovská 83, 18675 Prague 8, Czech Republic}
\email{jan.bok@matfyz.cuni.cz}
\author[F. Foucaud]{Florent Foucaud}
\address[F. Foucaud]{Université Clermont Auvergne, CNRS, Clermont Auvergne INP,
Mines Saint-Étienne, LIMOS, 63000 Clermont-Ferrand, France}
\email{florent.foucaud@uca.fr}
\author[D. A. Quiroz]{Daniel A. Quiroz}
\address[D. A. Quiroz]{Instituto de Ingeniería Matemática and Centro de Investigación y Modelamiento de Fenómenos Aleatorios - Valparaíso, CIMFAV, Universidad de Valparaíso, Valparaíso, Chile}
\email{daniel.quiroz@uv.cl}
\author[J.-F. Raymond]{Jean-Florent Raymond}
\address[J.-F.~Raymond]{Univ.\ Lyon, CNRS, ENS de Lyon, Université Claude Bernard Lyon 1, LIP UMR5668,
  Lyon, France}
\email{jean-florent.raymond@cnrs.fr}

\date{\today}

\begin{document}

\begin{abstract}
The \emph{$r$-neighbourhood complexity} of a graph $G$ is the function counting, for a given integer $k$, the largest possible number, over all vertex-subsets $A$ of size $k$, of subsets of $A$ realized as the intersection between the $r$-neighbourhood of some vertex and $A$. A~refinement of this notion is the \emph{$r$-profile complexity}, that counts the maximum number of distinct distance-vectors from any vertex to the vertices of $A$, ignoring distances larger than~$r$. Typically, in structured graph classes such as graphs of bounded VC-dimension or chordal graphs, these functions are bounded, leading to insights into their structural properties and efficient algorithms. 

We improve existing bounds on the $r$-profile complexity (and thus on the $r$-neighbourhood complexity) for graphs in several structured graph classes. We show that the $r$-profile complexity of graphs excluding $K_h$ as a minor is in  $O_h(r^{3h-3}k)$. For graphs of treewidth at most~$t$, we give a bound in $O_t(r^{t+1}k)$, which is tight up to a function of~$t$ as a factor. These bounds improve results of Joret and Rambaud and answer a question of their paper [Combinatorica, 2024]. We also apply our methods to other classes of bounded expansion such as graphs excluding a fixed complete graph as a subdivision.

For outerplanar graphs, we can improve our treewidth bound by a factor of $r$ and conjecture that a similar improvement holds for graphs with bounded simple treewidth.
For graphs of treelength at most~$\ell$, we give the upper bound of $O(k(r^2(\ell+1)^k))$, which we improve to $O\left (k\cdot (r 2^k + r^2k^2) \right)$ in the case of chordal graphs and $O(k^2r)$ for interval graphs.

Our bounds also imply relations between the order, diameter and metric dimension of graphs in these classes, improving results from [Beaudou et al., SIDMA 2017].
\end{abstract}

\maketitle
\section{Introduction}
\label{sec:intro}

An important structural property of a graph or hypergraph is the way its neighbourhoods are structured. A prominent parameter measuring this aspect is the \emph{Vapnik-Chervonenkis Dimension}, or VC-dimension, of a graph~\cite{VCdim-graph} or a hypergraph~\cite{VCdim}. By the Perles-Sauer-Shelah Lemma~\cite{Sauer72,Shelah72}, for a graph $G$ of VC-dimension at most $c$, the number of distinct intersections within any set $A$ of vertices and the neighbourhood of any vertex of $G$ is in $O(|A|^c)$, instead of $2^{|A|}$. This has led to the definition of the \emph{neighbourhood complexity} of a graph (also called \emph{trace function} or \emph{shatter function} in the context of hypergraphs~\cite{AMS19,F83}), and more generally, for any integer $r\geq 1$, the \emph{$r$-neighbourhood complexity} of a graph. Informally speaking, this is the function assigning to an integer $k$ the maximum number of distinct subsets of any vertex set $A$ of size $k$ that are realised as the $r$-neighbourhood (within~$A$) of some vertex of $G$.

The VC-dimension was originally defined for hypergraphs in the context of machine learning~\cite{VCdim}, by taking the hyperedges instead of the neighbourhoods. When the dataset (seen as a hypergraph) has bounded VC-dimension, and thus polynomial neighbourhood complexity, there is, for example, an efficient algorithm for the PAC-learning problem~\cite{DBLP:journals/jacm/BlumerEHW89}. For graphs, having bounded VC-dimension also has important algorithmic applications, see for example~\cite{DBLP:journals/jacm/BonamyBBCGKRST21,BLLPT15,DBLP:journals/siamcomp/DucoffeHV22,Mustafa2017}. Graph classes whose members have bounded VC-dimension (and thus neighbourhood complexity polynomial in $k$) include for example sparse graphs (such as graphs of bounded degeneracy), geometric intersection graphs (e.g.\ interval graphs, line graphs, or disk graphs), graphs with no 4-cycles, and structured dense graphs (for instance, graphs of bounded clique-width or twin-width)~\cite{BLLPT15}. More generally, many of these types of graphs also have their $r$-neighbourhood complexity polynomial in $k$ for any integer $r\geq 1$. This is the case for example for graphs in classes that have bounded expansion~\cite{Reidl19} or are nowhere dense~\cite{Eickmeyer17}.

Our goal is to improve known upper bounds on the $r$-neighbourhood complexity and the related $r$-profile complexity of graphs in structured graph classes. We next formally define these notions.

\subsection*{Neighbourhood and profile complexity}
Let us formally define neighbourhood and profile complexity.
\begin{definition}[neighbourhood complexity, $\nc_r$, $N_r$]\label{def:nc}
The \emph{$r$-neighbourhood complexity} is the function defined, for a graph $G$ and a positive integer $k$, by:
\[
\nc_r(G,k)=\underset{A\in \binom{V(G)}{k}}{\max} \big|\{N_r[v]\cap A, v\in V(G)\}\big|,
\]
where $N_r[v]$ denotes the \emph{$r$-neighbourhood} of $v$, i.e., the set of vertices of $G$ at distance at most~$r$ from $v$.
\end{definition}

It turns out that the $r$-neighbourhood complexity of sparse graphs is linear in $k$ for every~$r$. More precisely, Reidl, S{\'{a}}nchez Villaamil, and Stavropoulos proved in~\cite{Reidl19} that a graph class $\mathcal C$ that is closed under taking subgraphs has \emph{bounded expansion} if and only if there is a function $f$ such that for every $G \in \mathcal C$ and $r,k\in\mathbb{N}$, $\nc_r(G,k)\leq f(r)\cdot k$.
(Bounded expansion is a very general notion of sparsity that includes graph classes excluding a fixed graph as a minor or subdivision, and thus classes whose members have bounded treewidth, and classes whose members have bounded maximum degree, among others.)
More generally, nowhere dense classes of graphs have near-linear neighbourhood complexity~\cite{Eickmeyer17}. We refer to the book~\cite{sparsity} for more details on bounded expansion and nowhere dense graph classes.

Neighbourhood complexity plays an important role in recent advances of the study of structured dense graphs via the notions of merge-width and structurally sparse graph classes~\cite{BG25,DBLP:journals/corr/abs-2601-14906,dreier2026neighborhoodcomplexityradius1mergewidth}.

The neighbourhood complexity is not only an important structural graph measure, but it also has algorithmic applications. For example, in kernelization (a subarea of parameterized complexity) for problems related to distances such as $r$-\textsc{Dominating Set} or $r$-\textsc{Independent Set}, bounding the number of $r$-neighbourhoods may allow to discard vertices behaving the same way and keep only one representative for each neighbourhood. For algorithmic use of neighbourhood complexity, see for instance~\cite{Eickmeyer17,GAJARSKY} or the discussion in the introduction of~\cite{Reidl19}. Bounds for specific values of $r$ (typically $r=1$) have also been used to design algorithms, see e.g.~\cite{berthe2024subexponential,DBLP:conf/focs/DreierEMMPT24,lokshtanov2022subexponential}.

A useful refinement of the intersection of the $r$-neighbourhood of a vertex with a given set is the vector of $r$-truncated distances to vertices of this set, called \emph{(distance) profile} and formally defined as follows.\footnote{Note that this definition is different from the ``distance profiles'' recently used in~\cite{DBLP:conf/icalp/KlukPPS25}.}
\begin{definition}[profile, $p_r$, $\cp_r$]\label{def:prov}
 Given a graph $G$, a set $A$ of its vertices, and an integer $r$, we define the {\em $r$-profile} of some vertex $v$ of $G$ with respect to $A$ as the function
 \[
 p_r(v, A) : \left\{ \begin{array}{lcc} A & \rightarrow &\mathbb{N} \\ a & \mapsto &\cp_r(\dist(a,v))
 \end{array}\right.
 \] 
 where $\cp_r(\ell)=\ell$ if $\ell\leq r$, and $\cp_r(\ell)=+\infty$ otherwise.
\end{definition}
 
When studying the $r$-neighbourhood complexity, several authors (see~\cite{Eickmeyer17,JR}) have used this refinement and defined the notion of \emph{profile complexity}.

\begin{definition}[profile complexity, $\pc_r$]\label{def:prog}
The \emph{$r$-profile complexity} of a graph $G$ is the function defined by: 
\[
\pc_r(G,k)=\underset{A\in \binom{V(G)}{k}}{\max} \pc_r(G,A),
\]
where $\pc_r(G,A)$ counts the number of different $r$-profiles with respect to the subset $A$ except the ``all $+\infty$'' profile\footnote{We exclude this profile because it is more convenient in the proofs.}, i.e. $\pc_r(G,A)=\big|\{p_r(v,A), v\in V(G)\} \setminus \{a \mapsto +\infty \}\big|$.
\end{definition}

For a graph $G$ and two integers $k$ and $r$, one can check that we always have $\nc_r(G,k)\leq \pc_r(G,k)+1$~\cite{Eickmeyer17,JR}. Moreover, Joret and Rambaud proved the following lemma, further tightening the relation between these two functions for many natural classes of graphs.

\begin{lemma}[{\cite[Lemma 8]{JR}}]\label{lemma:relation-NC-PC}
Let $\mathcal{C}$ be a graph class stable by the operation of adding pendant vertices (i.e., whenever one attaches a new vertex of degree one to any vertex of a graph in $\mathcal{C}$, the resulting graph is in $\mathcal{C}$ as well). If there exists a function $f_{\mathcal{C}}: \mathbb{N}^2 \rightarrow \mathbb{N}$ such that for any graph $G$ in $\mathcal{C}$, and any integers $r$ and $k$ we have $\nc_r(G,k)\leq f_{\mathcal{C}}(r,k)$, then for any $G$, $r$, $k$ as above, $\pc_r(G,k)\leq f_{\mathcal{C}}(r,(r+1)k)$.
\end{lemma}

Lemma~\ref{lemma:relation-NC-PC} implies that for any graph class $\mathcal{C}$ satisfying the mild condition of being closed under adding pendant vertices, any upper bound obtained on the neighbourhood complexity can be extended to an upper bound on the profile complexity. In the specific case when the dependency on $k$ is linear (meaning $f: (r,k) \mapsto g(r)\cdot k$ for some function $g : \mathbb{N} \rightarrow \mathbb{N}$), the ratio between both bounds is $r+1$. Joret and Rambaud~\cite[Corollary 37]{JR} use this fact in the context of sparse graphs to extend lower bounds on the profile complexity to the neighbourhood complexity of such graphs. 

The notions of neighbourhood complexity and profile complexity are closely related to \emph{graph identification problems}, as we will see now.

\subsection*{Connection to Metric Dimension and other identification problems} In the area of identification problems, one wishes to distinguish the elements of a graph or discrete structure by the means of a small substructure.

A prominent example in this area is the concept of a \emph{resolving set} of a graph $G$, which is a set $S$ of vertices such that for any two distinct vertices $u,v$ of $G$, there is a vertex in $S$ with $\dist(u,s)\neq\dist(v,s)$. In other words, the $\diam(G)$-profiles of the vertices of $G$ with respect to $S$, are all distinct, where $\diam(G)$ is the diameter of $G$.

The \emph{metric dimension} $\md(G)$ of a graph $G$ is the smallest size of a resolving set of $G$. The concept was introduced in the 1970s~\cite{Harary76,S75} and extensively studied since then, with applications such as detection problems in networks, graph isomorphism, coin-weighing problems, or machine learning; see the surveys~\cite{KuziakYeroSurvey,TillquistSurvey}. As shown by Joret and Rambaud~\cite{JR}, the notion of $r$-profile complexity is closely connected to the study of the metric dimension of a graph with bounded diameter. Indeed, for a graph $G$ and any integer~$k$, if $\pc_r(G,k)\leq f(r,k)$ for some function $f$, then $G$ has at most $f(\diam(G),\md(G))$ vertices, since every vertex in $G$ needs a distinct $\diam(G)$-profile with respect to any optimal resolving set $S$. The question of finding the best possible upper bound on the number of vertices of a graph as a function of the metric dimension and diameter was studied for various graph classes in~\cite{LFetal,FMNPV17,DBLP:journals/combinatorics/HernandoMPSW10}. In fact, as attested by our work, the methods from these papers are at times applicable to the $r$-profile complexity as well. We note that the variant of \emph{$r$-truncated metric dimension}, where two vertices $u,v$ are only resolved by $s$ if at least one of them is at distance at most $r$ from $s$, is also studied~\cite{DBLP:journals/dam/FrongilloGLTY22} and corresponds more closely to $r$-profile complexity.

As observed in~\cite{DBLP:journals/ejc/BonnetFLP24}, a similar connection for the $r$-neighbourhood complexity exists with another type of identification problem, namely the concept of an \emph{$r$-locating-dominating set}, which is a set $S$ of vertices of a graph $G$ such that for every vertex in $V(G)\setminus S$, the intersection between its $r$-neighbourhood and $S$ is non-empty and unique~\cite{DBLP:journals/ejc/Honkala09}. This concept (and the close variant of $r$-identifying codes) is widely studied, with many applications. We refer to the book chapter~\cite{lobstein2020locating} and the extensive online bibliography maintained at~\cite{lobstein2012watching} for more on the vast literature of identification problems. Clearly, if $\nc_r(G,k)\leq f(r,k)$ for some function $f$, then $G$ has at most $f(r,\ld_r(G))$ vertices, where $\ld_r(G)$ is the smallest size of an $r$-locating-dominating set of $G$. The question of finding the best possible upper bound on the number of vertices of a graph $G$ as a function of $r$ and $\ld_r(G)$ was studied for various classes of graphs, especially for $r=1$, see~\cite{DBLP:journals/fuin/ChakrabortyFPW24,DBLP:journals/jgt/FoucaudGNPV13,FMNPV17,RS84,DBLP:journals/networks/Slater87}.

\subsection*{Previous work}

Improving in particular a previous bound of Sokolowski~\cite{DBLP:journals/combinatorics/Sokolowski23} for planar graphs, Joret and Rambaud in~\cite{JR} showed the following upper bounds on the profile complexity (and thus, neighbourhood complexity) of several classes of sparse graphs:

\begin{theorem}[\cite{JR}]\label{thm:jrall}
For every graph G,
\begin{enumerate}
    \item\label{e:tw} $\pc_r(G,k)\in O\left ((t+1)\binom{r+t}{t}r^{t}k\right)$ if $G$ has treewidth at most $t$;
    \item\label{e:minor} $\pc_r(G,k)\in O_h(r^{h^2-1}k)$ if $G$ excludes $K_h$ as a minor;
    \item $\pc_r(G,k)\in O_g(r^{5}k)$ if $G$ has Euler genus at most $g$;
    \item\label{e:planar} $\pc_r(G,k)\in O(r^4k)$ if $G$ is planar.
\end{enumerate}
\end{theorem}

Moreover, Joret and Rambaud~\cite[Theorem 36]{JR} described, for every positive integer $t$ and arbitrarily large values of $r$, graphs $G$ of treewidth at most $t$ with $\pc_r(G,k)\in \Omega_t(r^{t+1}k)$. Combining this with Lemma~\ref{lemma:relation-NC-PC}, they also obtain graphs $G$ of treewidth at most $t$ with $\nc_r(G,k)\in \Omega_t(r^{t}k)$~\cite[Corollary 37]{JR}.

As described by Joret and Rambaud~\cite{JR}, their bounds improve results from~\cite{LFetal} on the metric dimension. Let $G$ be a graph with order $n$, metric dimension $k$ and diameter $d$. As mentioned above, if $\pc_r(G,k)\leq f(r,k)$ for some function $f$, then $G$ has at most $f(d,k)$ vertices, and the above bounds on $\pc_r$ can be directly applied. It was proved in~\cite{DBLP:journals/combinatorics/HernandoMPSW10} that for any graph $G$, we have $n\in O(k(2d/3)^k)$, and this is asymptotically tight. It is known that $n\in O(dk^2)$ if $G$ is an interval graph or a permutation graph, and $n\in O(dk)$ if $G$ is a cograph, a unit interval graph, or a bipartite permutation graph, and these bounds are also tight~\cite{FMNPV17}. The bound $n\in O((dk)^{d\cdot 2^{O(w)}})$ holds if $G$ has rank-width at most~$w$~\cite{LFetal}. It is known that $n\in O(kd^2)$ if $G$ is outerplanar~\cite{LFetal}. Moreover, $n\leq kd^2(2\ell+1)^{3w+1}$ if $G$ has a tree-decomposition of width $w$ and length $\ell$~\cite{LFetal}, which implies that $n\in O(2^{2^{O(k)}}d^2)$ if $G$ is chordal.

Exact bounds for the case $r=1$ for graphs of given pathwidth and treewidth are provided in~\cite{berthe2026neighbourhoodcomplexityidentificationproblems}. The $r$-neighbourhood complexity for certain classes of dense structured graphs has also been studied, see for example~\cite{Paszke20} for graphs of bounded clique-width, \cite{DBLP:journals/ejc/BonnetFLP24} for graphs of bounded twin-width and~\cite{BG25} for graphs of bounded merge-width.

\subsection*{Our results}

In this paper we optimally improve bounds from the literature on the $r$-profile complexity of several classes of sparse graphs and give new bounds for chordal and similar tree-structured graphs. As discussed above, these results also provide improved bounds on the $r$-neighbourhood complexity and metric dimension of the considered graph classes.

We first turn our attention to graphs of bounded treewidth and prove the following.

\begin{restatable}{theorem}{thtw}\label{thm:TW}
Let $t,r$ be two positive integers and let $G$ be a graph of treewidth at most $t$. Then, $\pc_r(G,k)\in O(t^{O(t)}r^{t+1}k)$.
\end{restatable}

Note that Theorem~\ref{thm:TW} is asymptotically tight by a construction from~\cite[Theorem 36]{JR} providing (for every positive integer $t$ and arbitrarily large values of $r$) graphs $G$ of treewidth at most $t$ with $\pc_r(G,k)\in\Omega_t(r^{t+1}k)$.
While our proof uses the same \emph{guarding sets} technique as in \cite{JR}, we show that these sets can be constructed using the least common ancestor closure in the decomposition tree of the bags containing the vertices of the set with respect to which the profiles are counted. Compared to the generalised colouring numbers approach followed in~\cite{JR} to obtain the bound of Theorem~\ref{thm:jrall}.\eqref{e:tw}, our proof results in a smaller guarding set (albeit with slightly larger members), and consequently in better bounds. The use of the least common ancestor closure in this context is inspired by an algorithm in \cite{Fominetal} for a completely different problem (hitting planar minors).

Item~\eqref{e:planar} of Theorem~\ref{thm:jrall} by Joret and Rambaud gives an $O(r^4k)$ bound for the profile complexity of planar graphs. Meanwhile our Theorem~\ref{thm:TW} gives a $O(r^3k)$ for treewidth 2 graphs (which are planar). We improve this further for outerplanar graphs with the following bound, extending a result for metric dimension from~\cite{LFetal}.

\begin{restatable}{theorem}{thouter}\label{th:outer}
Let $G$ be an outerplanar graph, then $\pc_r(G,k) \in O(r^2k)$.
\end{restatable}

This bound is asymptotically tight (even for trees), see~\cite{LFetal}. We conjecture that this result can be extended to an $O(t^{O(t)}r^{t}k)$ bound for graphs with bounded simple treewidth (see Section~\ref{sec:discussion}).
We then consider the general case of graphs excluding a fixed minor.

\begin{theorem}\label{thm:minor}
Let $h\geq 4$ and $r$ be positive integers and let $G$ be a graph with no $K_h$ minor. Then, $\pc_r(G,k)\in O(h^{O(h)}r^{3h-3}k)$.
\end{theorem}

Theorem~\ref{thm:minor} is a substantial improvement over the previous bound from Joret and Rambaud~\cite{JR} (see item \eqref{e:minor} of Theorem~\ref{thm:jrall}) and positively answers one of their open questions. Our approach builds on the guarding sets used by Joret and Rambaud, which are constructed in terms of the weak colouring numbers of the graph. We use a mix of the strong and weak colouring numbers and capitalise on the fact that graphs excluding a fixed minor have orderings which give good (and different) upper bounds on these numbers. The members of our guarding set are substantially smaller, giving us this improved result. 

Our approach for Theorem~\ref{thm:minor} can be used in general for classes with bounded expansion. We discuss some applications in Section~\ref{sec:uniformorder}, such as graphs excluding a fixed subdivision and intersection graphs of balls in $\mathbb{R}^d$.

We then turn to graphs of bounded treelength, i.e., graphs that admit a tree-decomposition where every pair of vertices in the same bag are at distance at most some constant $\ell$ in the graph. In this direction, we show the following.

\begin{restatable}{theorem}{thlength}\label{th:length}
Let $G$ be a graph of treelength at most $\ell$.
Then $\pc_r(G,k) \in O\left (k\cdot (r^2(\ell+1)^k) \right)$.
\end{restatable}

Moreover, we show that the above bound is tight up to the multiplicative factor of $k$ (and constants), already when $\ell=2$. However, in the case of chordal graphs (which coincide with graphs of treelength~1), the above result can be improved as follows.

\begin{restatable}{theorem}{thchordal}\label{th:chordal}
Let $G$ be a chordal graph.
Then $\pc_r(G,k) \in O\left (k\cdot (r 2^k + r^2k^2) \right)$.
\end{restatable}

We show that this bound is nearly tight up to constants, the multiplicative factor of $k$, and a reduced exponent of $k/2$ in the exponential part of the bound. The bound precisely describes the respective contributions of $k$ and~$r$ to the profile complexity of chordal graphs. In particular, it reveals that the exponential contribution of $k$ is only factor of a linear function of $r$, unlike for graphs of treelength~2 or more.

For interval graphs, we revisit a proof of~\cite{LFetal} to obtain the following.

\begin{restatable}{theorem}{thinterval}\label{thm:intervalgraphs}
	Let $G$ be a connected interval graph. Then, $\pc_r(G,k) \in O(k^2r)$.
\end{restatable}

We note that the bound of Theorem~\ref{thm:intervalgraphs} is asymptotically tight by the example from \cite[Proposition~8]{FMNPV17} providing, for every two integers $k\geq 1$ and $d\geq 2$, an interval graph of diameter~$d$ that has a resolving set of size $k$ and $\Theta(k^2d)$ vertices (that hence all have a different $d$-profile).

Our results have the following consequences about the metric dimension of the considered graph classes.

\begin{corollary}
Let $G$ be an $n$-vertex graph with diameter $d$ and metric dimension $k$. 
\begin{enumerate}
    \item If $G$ has treewidth at most $t$, then $n\in O_t(d^{t+1}k)$.
    \item If $G$ excludes $K_h$ as a minor for an integer $h\ge 4$, then $n\in O_h(d^{3(h-1)}k)$.
    \item If $G$ is outerplanar, then $n \in O(d^2k)$.
    \item If $G$ is chordal, then $n\in O(k(2^kd+k^2d^2))$.
    \item If $G$ has a treelength at most $\ell$, then $n \in O\left (k\cdot (d^2(\ell+1)^k) \right)$.
\end{enumerate}
\end{corollary}

The two first bounds improve results from~\cite{JR} and the two last bounds improve those from~\cite{LFetal}. The bound for outerplanar graphs was already proved in \cite{LFetal} and Theorem~\ref{th:outer} allows us to easily recover it.
We note that the bound for treewidth is asymptotically tight, indeed the construction provided in~\cite[Theorem 36]{JR} actually gives (for every positive integer $t$ and arbitrarily large values of $d$) graphs of treewidth at most $t$, metric dimension $k$, diameter $O(d)$, and $\Omega_t(d^{t+1}k)$ vertices.

\subsection*{Outline} 
In Section~\ref{sec:prelim} we introduce the necessary definitions. We give bounds for graphs of bounded treewidth and outerplanar graphs in Sections \ref{sec:tw} and \ref{sec:stw} respectively. In Section~\ref{sec:minorfree} we discuss the case of graph classes excluding a fixed minor and graphs that admit uniform orderings for the generalized colouring numbers. Section~\ref{sec:chordal} is devoted to graphs of bounded treelength, including chordal and interval graphs. Finally, we conclude in Section~\ref{sec:discussion} with open questions.

\section{Preliminaries}
\label{sec:prelim}

In this paper all graphs are simple, loopless, and undirected.

\subsection{Distances, profiles, and guarding sets.}\label{subsec:gs}
The \emph{distance} between two vertices $u$ and $v$ of a graph $G$ is the minimum number of edges of a path starting at $u$ and ending at $v$. The \emph{distance} $\dist_G(X,Y)$ between two vertex sets $X$ and $Y$ of $G$ is the minimum distance in $G$ between a vertex of $X$ and a vertex of~$Y$. We drop the subscript when there is no ambiguity. The \emph{diameter} $\diam(G)$ of $G$ is the maximum distance between any two of its vertices.

Recall that the definitions of the neighbourhood complexity $\nc_r$,  $r$-profiles $p_r$, the profile complexity $\pc_r$, and the truncating function $\cp_r$ have been given above in Definitions \ref{def:nc}, \ref{def:prov}, and \ref{def:prog}.

In the course of proving the different items of Theorem~\ref{thm:jrall}, Joret and Rambaud introduced the concept of a \emph{guarding set}, that we define now.
Let $G$ be a graph, $A\subseteq V(G)$, and $r,p\in \mathbb{N}$. A family $\mathcal S \subseteq 2^{V(G)}$ is a \emph{$(r,p)$-guarding set} for $A$ if:
\begin{enumerate}
    \item $|S| \leq p$ for every $S \in \mathcal S$, and
    \item for every $v \in V(G)$, there exists $S \in \mathcal S$ such that $S$ intersects every path of length at most $r$ in $G$ from $v$ to a vertex of $A$ (if any).
\end{enumerate}
Guarding sets are interesting because of the following result.

\begin{lemma}[{\cite[Lemma~12]{JR}}]\label{lem:guarding}
Let $r,p$ be nonnegative integers, $G$  a graph,  and $\mathcal S$ an $(r,p)$-guarding set for $A\subseteq V(G)$. Suppose that for some non-decreasing function $f$ and every $A'\subseteq V(G)$,
\[
\pc_r(G,A') \le f(r,|A'|).
\]
Then
\[
\pc_r(G,A) \le f(r,p)|\mathcal S|.
\]
\end{lemma}

\subsection{Tree representations and widths.}
A \emph{tree representation} of a graph $G = (V,E)$ is a pair $\T = (T, \{T_v\}_{v\in V(G)})$ such that
$T$ is a tree,
for every $v\in V(G)$, $T_v$ is a subtree of $T$ (called \emph{model} of $v$), and
for every $u,v \in V(G)$, if $uv\in E(G)$ then $T_u$ and $T_v$ have a common node in $T$.
To avoid any possible confusion between the vertices of the graphs $T$ and $G$ that play different roles here, we will use the synonym \emph{node} to refer to a vertex of the tree of a tree representation.

For a vertex $t$ of $T$, we define $\beta_\T(t)$ as the set of all vertices $v \in V(G)$ such that $t \in V(T_v)$ and call this set the \emph{bag} at $t$. We drop the subscript when there is no ambiguity.
The \emph{width} of $\T$ is defined as $\max_{t\in V(T)} |\beta(t)|-1$. The \emph{treewidth} of $G$, denoted by $\tw(G)$, is the minimum width of a representation of $G$.\footnote{This is not the usual definition of treewidth but it can easily be checked that the two definitions are equivalent.} Similarly we can define the \emph{treelength} $\tl(G)$ of $G$, where the \emph{length} of a representation $(T, \{T_v\}_{v\in V(G)})$ is defined as $\max_{t\in V(T)} \max_{u,v\in \beta(t)} \dist_G(u,v)$. The length of a tree-decomposition was defined in~\cite{DBLP:journals/dm/DourisboureG07}, see also~\cite{DBLP:journals/combinatorica/BergerS24} for a recent characterization. \emph{Chordal graphs} can be defined as graphs that have treelength at most one. A tree representation of length 1 is called a \emph{chordal representation}. It has the property that two vertices of the graph are adjacent if and only if their models in this representation intersect.

\subsection{Least common ancestors}
For a rooted tree $T$ and $M\subseteq V(T)$, the \emph{least common ancestor closure (LCA-closure)}
is the set $M'$ of all the least common ancestors of two (possibly identical) vertices in $M$. Note that $M \subseteq M'$, and that the least common ancestors of two vertices in $M'$ is in $M'$.
We will need the following folklore lemma (see \cite[Lemma 1]{Fominetal} for a proof).

\begin{lemma}[{\cite[Lemma 1]{Fominetal}}]\label{lem:lca}
Let $T$ be a tree and, for $M\subseteq V(T)$, let $M'$ be the LCA-closure of $M$. Then we have $|M'|\le 2|M|$ and, for every component $C$ of $T\setminus M'$, $|N_T(C)|\le 2$.
\end{lemma}

\subsection{Generalised colouring numbers}

The following parameters were introduced by Kierstead and Yang~\cite{KY}. Consider $r\in \mathbb{N}$, a graph $G$, and a linear ordering $L$ of $V(G)$. We say that a vertex $u\in V(G)$ is \emph{weakly $r$-reachable} from $v\in V(G)$ if there exists an $uv$-path $P$ of length at most $r$ such that $u$ is minimum with respect to $L$ in $V(P)$. If we additionally have $v\le_L w$ for all vertices $w\in P\setminus\{u\}$,  we say that $u$ is \emph{strongly $r$-reachable} from $v$. Let $\WReach_r[G,L,v]$ and  $\SReach_r[G,L,v]$ be the sets of vertices that are weakly $r$-reachable and strongly $r$-reachable from $v$, respectively.  We set $$\wcol_r(G,L)=\max_{v\in V(G)} |\WReach_r[G,L,v]|,\phantom{space} \scol_r(G,L)=\max_{v\in V(G)} |\SReach_r[G,L,v]|,$$ and define the \emph{weak $r$-colouring number}, $\wcol_{r}(G)$, and the \emph{strong $r$-colouring number}, $\scol_{r}(G)$, of $G$, respectively, as follows:
	$$\wcol_{r}(G)=\min\limits_{L} \wcol_r(G,L), \phantom{space} \scol_{r}(G)=\min\limits_{L} \scol_r(G,L).$$

\section{Graphs of bounded treewidth}
\label{sec:tw}

In this section we prove Theorem~\ref{thm:TW}, which we restate below for convenience.

\thtw*

The tightness of our bound is witnessed by Theorem 36 of \cite{JR}.
Our bound relies on a carefully chosen guarding set, given by the following lemma.

\begin{lemma}\label{lem:guardingtw}
Let $G$ be a graph with treewidth at most $t$. For every nonempty $A\subseteq V(G)$ there is an $(r,2(t+1))$-guarding set for $A$ in $G$ of size at most $2|A|$. 
\end{lemma}
\begin{proof}
We take a tree representation $(T, \{T_v\}_{v\in V(G)})$
of width at most $t$ of $G$, and root $T$ at an arbitrary node $s$. Note that we may assume that the tree representation is such that $|\beta(u)\cap \beta(v)|\le t$ for every $uv\in E(T)$ as otherwise $\beta(u)=\beta(v)$ and we could identify these two nodes without changing the width of the representation. For every vertex $x\in V(G)$, we let $s_x$ be the (unique) node of $T_x$ that is the closest to the root $s$.

We let $B=\{s_x \mid x\in A\}$ and $B'$ be the LCA-closure of $B$ in $T$ with root $s$. For every node $b\in B'$, we let $p'(b)$ be the first vertex of $B'\setminus \{b\}$ met along the path from $b$ to the root $s$, or $p'(b)=b$ if no such node exists. We let $\mathcal{S}=\{\beta(b)\cup \beta(p'(b)) \mid b\in B'\}$. We clearly have $|\mathcal{S}|\le |B'|$, and thus by Lemma~\ref{lem:lca} we have 
\begin{equation}\label{eq:sizeguardset}
    |\mathcal{S}|\le 2|B|\le 2|A|.
\end{equation}
Moreover, for every component $C$ of $T\setminus B'$, the set $N_T(C)$ is of the
form $\{b\}$ or $\{b, p'(b)\}$ for some $b\in B'$. Indeed, we have $|N_T(C)|\le 2$
by Lemma~\ref{lem:lca}, while $N_T(C)\neq\varnothing$ since $T$ is connected and
$B'\neq\varnothing$.
If $N_T(C)=\{b\}$ for some $b\in B'$, then
the claim follows immediately. Hence, we may
assume that $N_T(C)=\{b_1,b_2\}$ for two
distinct vertices $b_1,b_2\in B'$. Let $w$ be the least common ancestor of $b_1$ and $b_2$ in
$T$. The $b_1b_2$-path of $T$ contains $w$ and all its internal nodes lie in $C$,
while $w\in B'$ as $B'$ is LCA-closed. Hence $w\in\{b_1,b_2\}$, say $w=b_2$, and
$b_2$ is therefore an ancestor of $b_1$. As moreover no internal node of the
$b_1b_2$-path lies in $B'$, the node $b_2$ is the first node of
$B'\setminus\{b_1\}$ met along the path from $b_1$ to the root $s$, that is,
$b_2=p'(b_1)$. The definition of $\mathcal{S}$ and the fact that the tree representation has width at most $t$ imply the following:
\begin{equation}\label{eq:sizeguard}
    \text{for every}\ S\in\mathcal{S},\ |S|\le 2(t+1).
\end{equation}

By Equations \eqref{eq:sizeguardset} and \eqref{eq:sizeguard} all we have to show now is that for every $x\in V(G)$, there exists $S\in \mathcal{S}$ which intersects every $(x,A)$-path of length at most $r$ in $G$. We assume that there is such an $(x,A)$-path, as otherwise there is nothing to prove. If there is a node $b\in B'$ such that  $x\in \beta(b)$, then, trivially, the set $\beta(b)\cup \beta(p'(b))\in\mathcal S$ contains $x$ and hence intersects every $(x,A)$-path. So, we assume there is no such node, and thus $T_x$ is contained in some component $C$ of $T\setminus B'$. As mentioned earlier, $N_T(C)$ is of the form $\{b\}$ or $\{b, p'(b)\}$ for some $b\in B'$.

Let $a$ be some vertex in $A$ and first assume that there  is some node $c\in V(C)$ such that $c \in V(T_a)$. Since we have $s_a\in B$ and $T_a$ being connected, then at least one of $b$ and $p'(b)$ belongs to $T_a$. Thus $a\in \beta(b)$ or $a\in \beta(p'(b))$.
In particular the set $\beta(b)\cup \beta(p'(b))\in \mathcal S$ intersects every $(x,a)$-path, as claimed. Let us now deal with the remaining case where $T_a$ does not have a node in
common with $C$. Recall that $T_x$ is contained in $C$ and that $N_T(C)$ is of
the form $\{b\}$ or $\{b,p'(b)\}$, so $C$ is a component of
$T\setminus\{b,p'(b)\}$. Using the property of tree representations that bags
are vertex separators, we deduce that the set $\beta(b)\cup \beta(p'(b))$
intersects every $(x,a)$-path of $G$.
We conclude that for every $a\in A$, the set $\beta(b)\cup \beta(p'(b))\in \mathcal S$ intersects every $(x,a)$-path.
So $\mathcal{S}$ is indeed the claimed guarding set.
\end{proof}

Bousquet and Thomassé~\cite{DBLP:journals/dm/BousquetT15} defined the \emph{distance VC-dimension} of a graph $G$ as follows (see also Chepoi, Estellon and Vaxès~\cite{DBLP:journals/dcg/ChepoiEV07} who used this notion earlier, yet without naming it). Let $\mathcal{H}$ be the hypergraph with vertex set $V(G)$ and having $\{N_r[v]\mid v\in G,\,\, r\ge 0\}$ as its edge set. The distance $VC$-dimension of $G$ is the VC-dimension of $\mathcal{H}$. Bousquet and Thomassé showed that graphs excluding $K_h$ as a minor have distance VC-dimension at most $h-1$. 
Beaudou et al.~\cite{LFetal} used this result, together with the Perles-Sauer-Shelah Lemma~\cite{Sauer72,Shelah72} to bound the metric dimension of graphs excluding a fixed complete minor. Following similar steps, Joret and Rambaud~\cite{JR} obtained the following result, which will be useful for us.

\begin{theorem}[\cite{JR}]\label{thm:ktminor}
Let $t\ge 3$ be an integer and $G$ a graph with no $K_t$ minor. Then for every nonempty set $A\subseteq V(G)$ and every $r\ge 0$, we have
\[
\pc_r(G,A) \le (r+1)^{t-1}|A|^{t-1}.
\]
\end{theorem}

We are now ready to prove Theorem~\ref{thm:TW}.
\begin{proof}[Proof of Theorem~\ref{thm:TW}.]
We will actually show the following more accurate bound: for every graph $G$ with treewidth at most $t$, every subset $A\subseteq V(G)$ and every integer $r\geq 0$,
\[
\pc_r(G,A)\leq 2^{t+2}(r+1)^{t+1}(t+1)^{t+1} \cdot |A|.
\]
We may clearly assume that $A\neq\varnothing$, as $\pc_r(G,\varnothing)=0$.
Let $\mathcal{S}$ be a $(r,2(t+1))$-guarding set for $A$ in $G$ of size at most $2|A|$ as given by Lemma~\ref{lem:guardingtw}.
Since graphs with treewidth at most $t$ exclude $K_{t+2}$ as a minor, by 
Theorem~\ref{thm:ktminor} we have for every $A'\subseteq V(G)$ the bound $\pc_r(G,A') \le (r+1)^{t+1}|A'|^{t+1}$.
Therefore, by Lemma~\ref{lem:guarding}  we obtain
\begin{equation*}
\begin{split}
  \pc_r(G,A) & \le (r+1)^{t+1}(2(t+1))^{t+1}|\mathcal S|\\
  & \le 2^{t+2}(r+1)^{t+1}(t+1)^{t+1}|A|. \qedhere
  \end{split}
\end{equation*}
\end{proof}

\section{Outerplanar graphs}\label{sec:stw}

By Theorem~\ref{thm:TW}, if $G$ has treewidth at most~2, then its profile complexity is in $O(r^{3}|A|)$, and this is tight (up to a constant factor) as noted in the introduction. In this section, we extend a result from~\cite{LFetal} about the metric dimension, to show that the profile complexity of outerplanar graphs is in $O(r^{2}|A|)$. As mentioned in the introduction, we conjecture that a similar improvement is possible for all graphs with bounded simple treewidth (see Conjecture~\ref{conj:simpleTW} in Section~\ref{sec:discussion}). Proposition~20 of~\cite{LFetal} implies that the following is asymptotically tight (even for trees, see Theorem~2 in~\cite{LFetal}).

\thouter*

\begin{proof}
Actually we show the following stronger statement:
For every outerplanar graph $G$, every subset $A\subseteq V(G)$, and every integer $r\ge 0$, we have
\[
\pc_r(G,A) \le 1+(2r+2)^2|A|.
\]
In the following we assume that $G$ is connected. To deal with disconnected graphs, we can apply separately the argument to every component that contains a vertex of $A$ and sum the numbers of profiles in each. Since $A$ contributes linearly to the bound this will give the desired result.

Following the proof of Theorem 19 of~\cite{LFetal}, we use the fact that outerplanar graphs can be represented in the plane in such a way that all vertices lie in the boundary of a circle, with the edges lying inside the circle (see~\cite{Six1999}). Let $<$ be an ordering of $V(G)$ obtained by moving along this circle, starting at some vertex $a_1 \in A$.

Let $d$ be the diameter of $G$. For each $1\le i\le d$, let $L_i$ be the set of vertices at distance exactly $i$ from $a_1$. The following result from~\cite{LFetal} is essential to the proof.

\begin{lemma}[{\cite[Claim 19.B]{LFetal}}]\label{lem:intervals}
    Let $i\in\{1,\dots, d\}$, $a\in A$, and let $y$ be a vertex of $L_i$ which minimizes the distance to vertex $a$ among all vertices of $L_i$. For $u,v\in L_i$, if $y<u<v$ or $v<u<y$, then we have $\dist(a,u)\le \dist(a,v)$.
\end{lemma}

For every $i\in\{1,\dots, d\}$, let $A_i$ be the subset of those vertices of $A$ whose $r$-neighbourhood intersects $L_i$. Observe that for every $a\in A$ there are at most $2r+1$ values of $i$ such that $N_r[a]$ intersects $L_i$ (otherwise there would be a shortcut contradicting the definition of the $L_i$'s).
So we have
\begin{equation}\label{eq:sumlevels}
    \sum_{i=1}^{d}|A_i|\le (2r+1)|A|.
\end{equation}

Let us call \emph{interval} a set of vertices of $L_i$ that are consecutive in the ordering $<$. Fix $a\in A_i$, let $y$ be as in Lemma~\ref{lem:intervals}, and group the vertices of $L_i$ according to the value of $\cp_r(\dist(\cdot,a))$. By Lemma~\ref{lem:intervals} this value is non-decreasing when moving away from $y$ along $<$, in either direction. Hence each of the at most $r+2$ values it takes on $L_i$ is attained on at most two intervals, one on each side of $y$, the two intervals corresponding to the value at $y$ forming a single one. This gives a partition of $L_i$ into at most $2r+3$ intervals.

Since all these partitions, for $a\in A_i$, are partitions into intervals of the same ordering $<$, their common refinement is a partition of $L_i$ into at most $(2r+3)|A_i|$ intervals. The $r$-profile with respect to $A$ of a vertex $v\in L_i$ is determined by the interval containing $v$: indeed, $\dist(v,a)>r$ for every $a\in A\setminus A_i$ by definition of $A_i$, and $\cp_r(\dist(v,a))$ is constant on each interval for every $a\in A_i$. Hence $L_i$ contributes at most $(2r+3)|A_i|$ distinct $r$-profiles to $\pc_r(G,A)$. This also holds when $A_i=\varnothing$: in that case $a_1\notin A_i$, hence $i>r$ and every vertex of $L_i$ has the all-$+\infty$ profile, which is not counted in $\pc_r(G,A)$. Thus (also counting the profile of $a_1$) we have
\begin{align*}
\pc_r(G,A)  & \le 1+ \sum_{i=1}^{d}(2r+3)|A_i|\\
    &\le 1+(2r+3)(2r+1)|A|\\
    &\le 1+(2r+2)^2|A|,
\end{align*}
where the second inequality comes from~\eqref{eq:sumlevels}.
\end{proof}

\section{Graph classes excluding a fixed minor and other classes}\label{sec:minorfree}

In this section, we give an affirmative answer to Problem 42 from~\cite{JR} by showing that if $G$ excludes $K_h$ as a minor, then, for every integer $r$ we have $\pc_r(G,k)\in h^{O(h)}(r+1)^{3(h-1)}k$. The used methods also apply to other types of graphs.

Using Lemma~\ref{lem:guarding}, Theorem~\ref{thm:ktminor}, and arguments from the paper of Reidl, S\'anchez Villaamil and Stavropoulos~\cite{Reidl19}, Joret and Rambaud~\cite{JR} proved that if $G$ excludes $K_h$ as a minor, we have $\pc_r(G,k)\in O(h^{O(h)}r^{h^2}k).$ As part of this proof, they construct guarding sets which were created in terms of the weak colouring numbers of $G$. We improve on these guarding sets by also considering the strong colouring numbers of $G$; the resulting guarding sets will have much smaller elements.

\begin{theorem}\label{thm:guardingminors}
Let $r,t,\alpha,\beta$ be nonnegative integers, $G$ a graph, and~$L$ a linear ordering of $V(G)$ with $\scol_{2r}(G,L)\le \alpha$ and $\wcol_{r}(G,L)\le \beta$. For every $A\subseteq V(G)$ there is an $(r,\alpha)$-guarding set for $A$ in $G$ of size at most $\beta|A|$.
\end{theorem}
\begin{proof}
Let $B=\cup_{a\in A}\WReach_r[G,L,a],$ and for every $b\in B$ let $S_b=\SReach_{2r}[G,L,b]$. We will show that the set $\mathcal{S} = \{S_b \mid b \in B\}$ is the desired guarding set. 

Since by hypothesis we have $|\mathcal{S}|\le\beta|A|$ and $|S|\le\alpha$ for every $S\in\mathcal{S}$, all we need to show is that for every vertex $x\in V(G)$ there exists $S\in \mathcal{S}$ which intersects every $(x,A)$-path of length at most $r$ in $G$. Assume there is such a path $P$ joining $x$ to some $a\in A$ (otherwise there is nothing to prove), and let $$\mu(x)=\max_L(B\cap \WReach_r[G,L,x]).$$ Notice that $B\cap \WReach_r[G,L,x]$ is nonempty and thus $\mu(x)$ does exist: for $y=\min_LV(P)$, the subpath $P_1$ of $P$ from $x$ to $y$ witnesses $y\in\WReach_r[G,L,x]$, while the subpath $P_2$ from~$y$ to $a$ witnesses $y\in B$.  

We will show that $S_{\mu(x)}$ intersects $P$. If $y=\mu(x)$, then we are done, so we assume otherwise. Let $Q$ be a path witnessing that $\mu(x)\in \WReach_r[G,L,x]$, and let $w$ be the first vertex of $P_1$ (moving from $x$ to $y$)  which satisfies $w<_L\mu(x)$. Note that $w$ exists because, by the choice of $\mu(x)$ and since $y\in B\cap\WReach_r[G,L,x]$, we at least have $y<_L\mu(x)$. The concatenation of $Q$ with the subpath of $P_1$ from $x$ to $w$ forms a path that witnesses that $w\in \SReach_{2r}[G,L,\mu(x)]=S_{\mu(x)}$. The result follows.
\end{proof}

In the rest of the section we give applications of Theorem~\ref{thm:guardingminors}.

\subsection{Graphs excluding a clique minor}

Van den Heuvel, Ossona de Mendez, Quiroz, Rabinovich, and Siebertz \cite{gencol} gave near-optimal upper bounds for the weak and strong colouring numbers of graphs excluding $K_h$ as a minor. The proof for the weak colouring numbers provides an ordering which gives good bounds for all weak colouring numbers at once (a so-called \emph{uniform ordering}). Moreover, that same ordering gives even better upper bounds for (all) the strong colouring numbers. This is made particularly clear in a recent work of Cortés, Kumar, Moore, Ossona de Mendez and Quiroz \cite{subchromatic}. The following is implied by Lemma 4.9 and the proof of Lemma 4.8 from \cite{subchromatic}.

\begin{lemma}[\cite{subchromatic}]\label{lem:weakstrong}
Let $h\ge4$ be an integer and $G$ a graph excluding $K_h$ as a minor. There is a linear ordering $L$ of $V(G)$ such that for every integer $r\ge0$ we have $$\scol_{r}(G,L)\le (h-3)(h-1)(2r+1), \mbox{ and}$$  $$\wcol_r(G,L)\le \binom{r+h-2}{h-2}(h-3)(2r+1).$$
\end{lemma}

Now we have all the ingredients for the main result of this section.

\begin{theorem}\label{thm:profilesminors}
Let $r,h$ be nonnegative integers with $h\ge 4$, $G$ a graph excluding $K_h$ as a minor. We have
\[
\pc_r(G,k)\le 4^h(h-3)h^{2(h-1)}(r+1)^{3(h-1)}k.
\]
In particular $\pc_r(G,k) \in h^{O(h)}(r+1)^{3(h-1)}k$.
\end{theorem}
\begin{proof}
Let $A$ be a $k$-subset of $V(G)$. Since $G$ excludes $K_h$ as a minor, by Lemma~\ref{lem:weakstrong}, $V(G)$ has an ordering $L$ such that
\begin{align*}
 \scol_{2r}(G,L) &\le (h-3)(h-1)(4r+1)\\
    &\le h^2(4r+1)\\
 \text{and}\quad \wcol_r(G,L) &\le \binom{r+h-2}{h-2}(h-3)(2r+1)\\
    &\le2(h-3)(r+1)^{h-1}.
\end{align*}
Thus, by Theorem~\ref{thm:guardingminors}, there exists an $(r,h^2(4r+1))$-guarding set $\mathcal{S}$ for $A$ in $G$ of size at most $2(h-3)(r+1)^{h-1}k$. Then by Theorem~\ref{thm:ktminor} and Lemma~\ref{lem:guarding} we obtain
\begin{equation*}
  \begin{split}
\pc_r(G,k) & \le (r+1)^{h-1}(h^2(4r+1))^{h-1}|\mathcal{S}| \\
& \le (r+1)^{h-1}(h^2(4r+1))^{h-1}\cdot 2(h-3)(r+1)^{h-1}k\\
 & \le 2(h-3)(r+1)^{2(h-1)}h^{2(h-1)}(4r+1)^{h-1}k \\
  & \le 4^h(h-3)h^{2(h-1)}(r+1)^{3(h-1)}k.\qedhere
\end{split}  
\end{equation*}
\end{proof}

\subsection{Uniform orderings in other classes with bounded expansion}\label{sec:uniformorder}

Our bounds for graphs excluding a complete graph as a minor spring from the fact that these graphs have uniform orderings for the generalised colouring numbers. While, in some sense, every class with bounded expansion allows for uniform orderings~\cite{VdHK21}, Theorem~\ref{thm:guardingminors} performs best with orderings that give substantially different bounds for weak and strong colouring numbers. In this section we consider two examples of graph classes with such orderings.

Kreutzer, Pilipczuk, Rabinovich and Siebertz~\cite{KPRS16} first proved that graphs excluding a fixed subdivision have uniform orderings for the colouring numbers. Their bounds were recently improved in the survey of Siebertz~\cite{Siebertz25}, and the following can be deduced from the proof of \cite[Corollary 4.27]{Siebertz25}.

\begin{lemma}[\cite{Siebertz25}]\label{lem:subdiv}
Let $s\ge 1$ be an integer and $G$ a graph excluding $K_s$ as a subdivision. There is a linear ordering $L$ of $V(G)$ such that for every integer $r\ge0$ we have $$\scol_{r}(G,L)\le s^2+(s^2+2s^{4r})+4s^4(2r+1), \mbox{ and}$$ $$\wcol_r(G,L)\le (4s^4(2r+2))^r.$$
\end{lemma}

We use this, together with Theorem~\ref{thm:guardingminors}, to obtain the following bounds on the neighbourhood complexity of graphs excluding a fixed subdivision. Note that Theorem~\ref{thm:ktminor} is not available here since its proof uses the fact that the class of graphs excluding a fixed minor has bounded distance VC-dimension. However, this does not hold for graphs excluding a fixed subdivision (as explained in the conclusion of the current paper). We therefore use the trivial bound $\pc_r(G,A')\le (r+2)^{|A'|}$, which holds for every graph $G$ and every $A'\subseteq V(G)$.

\begin{theorem}\label{thm:subdiv}
Let $r,s$ be positive integers, and $G$ a graph excluding $K_s$ as a subdivision. We have
\[
\pc_r(G,k)\le 4^rs^{4r}(2r+2)^{2s^{4r}+4s^4(2r+2)+r}k\in r^{s^{O(r)}}k.
\]

\end{theorem}
\begin{proof}
Let $A$ be a $k$-subset of $V(G)$. Lemma~\ref{lem:subdiv} and Theorem~\ref{thm:guardingminors}, together tell us that there is an $(r,2s^{4r}+4s^4(2r+2))$-guarding set $\mathcal{S}$ for $A$ of size at most $(4s^4(2r+2))^r$. Thus, by using Lemma~\ref{lem:guarding} and the fact that $\pc_r(G,A')\le (r+2)^{|A'|}$ for every $A'\subseteq V(G)$, we obtain
\begin{equation*}
  \begin{split}
\pc_{r}(G,A) & \le (r+2)^{2s^{4r}+4s^4(2r+2)}|\mathcal{S}| \\
& \le (r+2)^{2s^{4r}+4s^4(2r+2)}(4s^4(2r+2))^r|A|\\
 & \le 4^rs^{4r}(2r+2)^{2s^{4r}+4s^4(2r+2)+r}k. \qedhere
\end{split}  
\end{equation*}

\end{proof}

In \cite{DPUY22}, Dvo\v{r}\'ak, Pek\'arek, Ueckerdt and Yuditsky studied intersection graphs of various types of subsets of $\mathbb{R}^d$, and showed that if these objects are ordered in a non-increasing manner according to their diameter, then this ordering gives good (and different!) upper bounds for the weak and strong colouring numbers. As another example of the bounds that can be achieved through these orderings and Theorem~\ref{thm:guardingminors}, we study here intersection graphs of balls in $\mathbb{R}^d$. Similar results can be obtained for intersection graphs of scaled and translated copies of the same centrally symmetric compact convex subsets of $\mathbb{R}^d$ (like axis-aligned hypercubes), intersection graphs of ball-like subsets of $\mathbb{R}^d$ (see \cite{DPUY22} for a definition), and intersection graphs of comparable axis-aligned boxes. For these graph classes, Theorem~\ref{thm:ktminor} is again not available, and we resort once more to the trivial bound $\pc_r(G,A')\le (r+2)^{|A'|}$.

Let $S$ be a set of balls in $\mathbb{R}^d$. Its \emph{intersection graph} is the graph with vertex set $S$ and an edge $uv$ if and only if $u\cap v\ne\varnothing$. For an integer $t\ge 1$, $S$ is \emph{$t$-thin} if every point of $\mathbb{R}^d$ is contained in the interior of at most $t$ balls of $S$. 

\begin{theorem}\label{thm:balls}
Let $t$ and $d$ be positive integers. Let $S$ be a $t$-thin finite set of balls in $\mathbb{R}^d$. Let $G$ be the intersection graph of $S$, and $L$ an ordering of $S$ such that $u\le_L v$ whenever $\diam(u)\ge \diam(v)$. There exists $r_0$ (depending on $d$) such that for every $r\ge r_0$ we have $$\pc_r(G,k)\le 4^d(r+2)^{d+t(2r+1)^d}t\lceil \log_2r \rceil \binom{r+2t+2}{2t+2}k\in 2^{t\cdot r^{O(d)}}k.$$
\end{theorem}
\begin{proof}
Let $A$ be some $k$-subset of $V(G).$ By \cite[Lemma 1]{DPUY22} the ordering $L$ satisfies $\scol_r(G,L)\le t(2r+1)^d$ for every $r$ and by \cite[Theorem 3]{DPUY22}
it satisfies $\wcol_r(G,L) \le t\lceil \log_2r \rceil (4r-1)^d\binom{r+2t+2}{2t+2}$ for $r$ large enough. Hence, by Theorem~\ref{thm:guardingminors}, there is an $(r,t(2r+1)^d)$-guarding set $\mathcal{S}$ for $A$ of size at most $t\lceil \log_2r \rceil (4r-1)^d\binom{r+2t+2}{2t+2}$. Thus, by using Lemma~\ref{lem:guarding} and the fact that $\pc_r(G,A')\le (r+2)^{|A'|}$ for every $A'\subseteq V(G)$, we obtain
\begin{equation*}
  \begin{split}
\pc_{r}(G,A) & \le (r+2)^{t(2r+1)^d}|\mathcal{S}| \\
& \le (r+2)^{t(2r+1)^d}t\lceil \log_2r \rceil (4r-1)^d\binom{r+2t+2}{2t+2}|A|\\
 & \le 4^d(r+2)^{d+t(2r+1)^d}t\lceil \log_2r \rceil \binom{r+2t+2}{2t+2}k. \qedhere
\end{split}  
\end{equation*}
\end{proof}

\section{Graphs of bounded treelength and relatives}\label{sec:chordal}

In this section we give bounds on the profile complexity of interval graphs, chordal graphs, and graphs of bounded treelength.

\subsection{Interval graphs} 
In this section we show Theorem~\ref{thm:intervalgraphs}, that we restate below. Note that the profile complexity of interval graphs is described by Theorem~\ref{thm:balls} with $d=1$, with the additional condition of $t$-thinness (which for interval graphs corresponds to bounding the clique number). Here without this restriction the dependency in $k$, that is linear in Theorem~\ref{thm:balls}, becomes quadratic. The proof closely follows that of \cite[Theorem 7]{FMNPV17}.

\thinterval*

\begin{proof}
Let $A \subseteq V(G)$ be a subset of vertices of $G$ of size $k$ and let $a_1,\ldots,a_k$ be the elements of $A$, that are also intervals of the real line. For each $i$ in $\{1,\ldots,k\}$, we define an ordered set $L^i = \{ x^i_1 > x^i_2 > \ldots > x^i_s \}$ in the following way. Let $x_1^i$ be the left endpoint of $a_i$. Assuming $x_j^i$ is defined, let $x^i_{j+1}$ be the smallest among all left endpoints of the intervals of $G$ that end strictly after $x_j^i$. We stop the process when we have $j+1 = r$, or if $x_j^i = x_{j+1}^i$. Note that a vertex corresponding to an interval whose right endpoint lies within $[x^i_{j+1},x^i_j ]$ is at distance exactly $j+1$ of $a_i$. 

We similarly define the ordered set $R^i = \{y^i_1 < y^i_2 < \ldots < y^i_{s'} \}$: $y^i_1$ is the right endpoint of $a_i$, $y_{j+1}^i$ is the largest right endpoint among all the interval of $G$ that start strictly before $y_j^i$ and again, we stop the process when we have $j+1 = r$, or if $y_j^i = y_{j+1}^i$. Also, a vertex corresponding to an interval whose left endpoint lies within $[y^i_{j+1},y^i_j ]$ is at distance exactly $j+1$ of $a_i$.

Now to all of the orders $L^i$, we add a special vertex and set $x^i_{r+1}$ to be the smallest left endpoint among all the intervals of $G$ unless it would be the case that $x^i_{r+1} = x^i_{r}$. We do the analogous process with $R^i$ for every $i$. Clearly, there is no interval whose right endpoint is smaller than $x^i_s$, and analogously, there is no interval whose left endpoint is larger than $y_{s'}^i$.

Note that intervals at distance 1 of $a_i$ in $G$ are exactly the intervals starting before $y_1^i$ and finishing after $x_1^i$. More generally, for any interval of $G$, its $r$-truncated distance to $a_i$ is uniquely determined by the position of its right endpoint in the ordered set $L^i$ and the position of its left endpoint in $R^i$. Moreover, the interval $I_s$ that defines the point $x_s^i$ of $L^i$ and the interval $I_{s'}$ that defines the point $y_{s'}^i$ of $R^i$ are at distance at least $s+s'-4$ from each other. Indeed, a shortest path from $I_s$ to $I_{s'}$ contains $a_i$ or a neighbour $J$ of $a_i$. In the best case, $J$ is the interval $[x_2^i,y_2^i]$ and then $\dist(I_s,I_{s'})=\dist(I_s,J)+\dist(J,I_{s'}) \leq s-2 + s'-2$. Therefore, we have $s+s' \leq r-1+2 + r-1 +2$ and $L^i \cup R^i$ contains at most $2r+2$ points.

Consider now the union of all the sets $L^i \cup R^i$. Each of these sets has at most $2r+6$ points and they all have two common points at the extremities. Hence, the union contains at most $|A|\cdot(2r+4) + 2$ distinct points on the real line and thus defines a natural partition $\mathcal P$ of $\mathbb R$ into at most $|A|\cdot(2r+4) + 1$ intervals (there is no need to count the intervals before and after the extremities since no interval can start or end there). The profile of any interval in $V(G)\setminus A$ is now uniquely determined by the positions of its endpoints in $\mathcal P$. Let $I \in V(G) \setminus A$. For a fixed $i$, by definition of the sets $L^i$, the interval $I$ cannot contain two points of $L^i$ and similarly, it cannot contain two points of $R^i$. Thus, $I$ contains at most $2|A|$ points of the union of all sets $L^i$ and $R^i$. Therefore, if $P$ denotes a part of $\mathcal P$, there are at most $2|A| + 1$ intervals with left endpoints in $P$ having unique profiles. In total, there are at most $(|A| \cdot (2r+4) + 1) \cdot (2|A| + 1)$ intervals with different profiles in $V(G) \setminus A$ and
$$\pc_r(G,A) \leq (|A| \cdot (2r+4) + 1) \cdot (2|A| + 1) + |A| = 4 |A|^2 r + 8 |A|^2 + 2 |A| r + 7 |A| + 1.\qedhere$$
\end{proof}

The above bound is tight up to a constant factor according to the following result. We reproduce the construction for completeness (with slightly modified notations), as it will be used later in another construction.

\begin{proposition}[{\cite[Proposition~8]{FMNPV17}}]\label{prop:IG-construction}
 Given any two integers $r\geq 2$ and $k\geq 2$ (even), there exists an interval graph $IG_{r,k}$ such that
 \[
 \pc_r(IG_{r,k},k) \in \Omega(k^2r).
 \]
 \end{proposition}
\begin{proof}
Let $L>k/2$. For $i\in \{1,\ldots,k/2\}$ and $j\in \{1,\ldots,r\}$, define the interval 
$I_{i,j}=](j-1)L+i,jL+1/2+i[$. The intervals $I_{i,j}$ for a fixed $i$ induce a path
with $r$ vertices.

Let $a_i=I_{i,1}$ and $b_i=I_{i,r}$ for $1\leq i \leq k/2$. Furthermore, let us have $A=\{a_i , 1\leq i\leq k/2\}$, $B=\{b_i , 1\leq i\leq k/2\}$, and $S=A\cup B$.

We add some intervals that do not influence the shortest paths between the intervals/vertices $I_{i,j}$ (in particular, the distances from $I_{i,j}$ to vertices in $S$ do not change). 
First note that all the intervals $I_{i,j}$ have the same length. Thus, there is a natural ordering of these intervals, defined by $I_{i,j}<I_{i',j'}$ if and only if $j<j'$ or $j=j'$ and $i<i'$. In particular, any set of $k/2$ intervals that are consecutive in this order do not contain two intervals $I_{i,j}$ and $I_{i',j'}$ with $i=i'$.

Consider any interval $J=I_{i,j}$ with $2\leq j\leq r-2$. We add $k/2+1$ intervals after the end of $J$ as follows. Consider the set $\{J_0<J_1<\cdots<J_{k/2}\}$ of the first $k/2+1$ intervals starting after the end of $J$. Note that $J_0$ and $J_{k/2}$ correspond to a pair of intervals $I_{i',j'}$, $I_{i',j''}$ whose first index $i'$ is the same. For each interval $J_s$ ($1\leq s\leq k/2$), add an interval $I_{i,j,s}$ starting between the end of $J$ and the start of $J_0$ and ending between the start of $J_{s-1}$ and the start of $J_s$. These new intervals are all ending before the end of $J_{k/2}$ and thus are not changing the shortest paths between the old intervals/vertices.

All the intervals/vertices added this way have distinct $r$-profiles to $S$. A vertex corresponding to $I_{i,j}$ has distance $j-1$ to $a_{i'}$ if $i\leq i'$ and distance $j$ to $a_{i'}$ otherwise, and distance $r-j$ to $b_{i'}$ if $i'\leq i$ and distance $r-j+1$ to $b_{i'}$ otherwise. For any $i,j$ with $2\leq j\leq r-2$, the vertices corresponding to $I_{i,j,s}$ (with $1\leq s\leq k/2$) all have the same distances to $A$ as $I_0$ but their distances to $B$ are the same as for the $k/2$ intervals $I_{i',j'}$ that follow $I_{i,j}$ in the ordering (including $I_{i,j}$).

Thus, there are in total $kr+(k/2+1)(r-2)k/2=\Omega(k^2r)$ vertices in this graph, all with distinct $r$-profiles.
\end{proof}

\subsection{Chordal graphs}

For chordal graphs, we will use the following lemma to deal with clique separators.

\begin{lemma}\label{lem:1cliguard}
Let $G$ be a graph, $r\in \N$, $A\subseteq V(G)$.
Let $X \subseteq V(G)$ and suppose there is a clique $S\subseteq V(G)$ such that every path from $X$ to $G-X$ intersects $S$.

Then the number of $r$-profiles with respect to $A$ of vertices from $X$ is at most $(r+2)2^{|A|}$.
\end{lemma}

When dealing with graphs of bounded treelength in Section~\ref{sec:tl} we will need a variant of Lemma~\ref{lem:1cliguard} where the separator consists of two sets of small diameter. Because the proofs are very similar, we will wait until we prove the (more complicated) variant Lemma~\ref{lem:2cliguard} before we explain how it can be modified to yield Lemma~\ref{lem:1cliguard}.

We note that the bound of Lemma~\ref{lem:1cliguard} above has the correct order of magnitude, as shown by the following construction. Start with a (split) graph consisting of a clique $K$ and add $2^{|K|}-1$ independent vertices, each neighbouring with a different nonempty subset of $K$. Finally, attach a new path of length $r$ to every such vertex. It is not hard to observe that with $S=A=K$, there are $(r+1)(2^{|A|}-1)$ distinct $r$-profiles in the graph. Hence, the contributions of $r$ and $|A|$ to the bound of Lemma~\ref{lem:1cliguard} cannot be substantially improved. It also gives the following lower bound.
\begin{lemma}
For any integers $k,r\geq 1$, there is a chordal graph $G$ such that $\pc_r(G,k) \geq (r+1)(2^k-1)$.
\end{lemma}

We are now ready to give an upper bound for chordal graphs.
 
\thchordal*

\begin{proof}
Let $\mathcal{T} = (T, \{T_v\}_{v \in V(G)})$ be a chordal representation of $G$.
We start as in the proof of Lemma~\ref{lem:guardingtw}. That is, we root $T$ at some node $s$ and for every vertex $x\in V(G)$ we define $s_x$ as the node of $T_x$ that is closest to the root. Let $B = \{s_a \mid
 a \in A\}$ and let $B'$ be the LCA-closure of $B$ in $T$ with root $s$. By Lemma~\ref{lem:lca}, $|B'| \leq 2|B| \leq 2|A|$ and every component $C$ of $T-B'$ has at most two neighbours in $B'$.
 So we can deal separately with the following types of vertices:
 \begin{enumerate}
     \item \emph{Vertices in $\beta(t)$ for some $t \in B'$.} Consider a vertex $v \in \beta(t)$. As $\mathcal{T}$ is a chordal representation, for every $a\in A$ we have
     \[
        \dist_G(a, \beta(t)) \leq \dist_G(v,a) \leq \dist_G(a, \beta(t)) + 1.
     \]
    The distances from $\beta(t)$ to the vertices of $A$ are fixed, so in total, the vertices in $\beta(t)$ may have at most $2^{|A|}$ different $r$-profiles with respect to~$A$.
     \item \emph{Vertices in the bags of the components of $T-B'$ with the same unique neighbour $t$ in~$B'$.} Let $C$ denote the set of nodes of the components of $T-B'$ that have $t$ as unique neighbour in $B'$.
    Let $X = \bigcup_{c \in C} \beta(c)$.
    We now apply Lemma~\ref{lem:1cliguard} to $G, r, A,X$ and with $S = \beta(t)$, which is a clique separator, since $\mathcal{T}$ is a chordal representation. 
    So the number of $r$-profiles to $A$ of the vertices in $X$ is at most $(r+2) 2^{|A|}$.
    \item \emph{Vertices in the bags of a component $C$ of $T-B'$ that have two neighbours $t$ and $t'$ in $B'$.} Let $X = \bigcup_{c \in V(C)} \beta(c) \setminus (\beta(t)\cup \beta(t'))$. So $X$ consists of those vertices $v$ of $G$ whose model $T_v$ in the representation is a subtree of $C$ (and in particular the model does neither contain $t$ nor $t'$).
    Let $P$ denote the path in $C$ linking the two (unique) neighbours of $t$ and $t'$ from $C$. 
    For every vertex $v\in X$, let us denote by $\Lambda(v)$ the set of those vertices $u\in X$ such that $T_u$ shares nodes with $P$ and, among such vertices, $\dist_G(u,v)$ is minimum. Let $\lambda(v)$ denote this distance.
    
    We claim that the models of the vertices of $\Lambda(v)$ have a common node. Indeed, either $\lambda(v) = 0$ and $\Lambda(v) = \{v\}$, or $T_v$ does not have common nodes with $P$. In the latter case, let $x\in V(P)$ be the neighbour of the component of $C-V(P)$ that contains $T_v$ and observe that every model of a vertex in $\Lambda(v)$ contains the node $x$.
    Therefore, the union of these models, which we call $T_{\Lambda(v)}$, is a subtree of $C$.

    A crucial observation now is that the $r$-profile of $v$ is uniquely determined by that of a (virtual) vertex with model $T_{\Lambda(v)}$ and $\lambda(v)$. Indeed for every $a\in A$ we have:
    \[
    \dist_G(v,a) = \lambda(v) + \min_{u\in \Lambda(v)} \dist_G(u,a).
    \]
    Assuming that $v$ does not have the ``all-$+\infty$'' profile, $\lambda(v)$ can take up to $r+1$ different values.
    
    On the other hand, it can easily be shown that for a  vertex $u$ whose model $T_u$ intersects $P$ (such as the virtual vertex with model $T_{\Lambda(v)}$), no shortest path to vertices of $A$ contains a vertex of $X$ whose model does not intersect $P$. Towards a contradiction, suppose that such a shortest path has a vertex $w\in X$ whose model does not intersect $P$. Let $w'$ and $w''$ respectively denote the first vertices that do not satisfy this property encountered when following the path from $w$ forwards and backwards. If $T_{w'}$ and $T_{w''}$ have a common node, then $w'w''$ is an edge (as $\mathcal{T}$ is a chordal representation) so the considered shortest path is not induced, a contradiction. Otherwise, by connecting them using vertices whose model intersects $P$ we can create an induced cycle of length at least 4, which is impossible in a chordal graph.
    
    Hence, the $r$-profile of such a vertex is uniquely determined by the intersection of the model with $P$, i.e. the subpath with nodes $V(T_u) \cap V(P)$. This reduces the problem to interval graphs. As shown in Theorem~\ref{thm:intervalgraphs}, there are $O(k^2r)$ profiles to $A$ for such vertices. Therefore, the vertices in $X$ have $O(k^2r^2)$ different $r$-profiles to $A$.
 \end{enumerate}

    It remains to sum things up. As noted above, $|B'| \leq 2|A|$ so there are at most $2|A|$ sets of vertices to consider in each of the following three cases detailed above. Hence, the number of $r$-profiles with respect to $A$ sums to $O\left (k\cdot (r 2^k + r^2k^2) \right)$, as claimed.
 \end{proof}

We next show that both the exponential contribution of $k$ times a linear factor in $r$ and the quadratic factor in both $r$ and $k$ are necessary in the above bound.

 \begin{proposition}\label{prop:chordal-construction}
 Given any two integers $r\geq 2$ and $k\geq 2$ such that $k$ is even, there exists a chordal graph $G_{r,k}$ such that
 \[
 \pc_r(G_{r,k},k) \in \Omega(r2^{k/2}+r^2k^2).
 \]
 \end{proposition}
 
 \begin{proof}
   The graph $G_{r,k}$ is obtained essentially by merging two copies of the construction described after Lemma~\ref{lem:1cliguard} with the interval graph construction from Proposition~\ref{prop:IG-construction} (and adding some pending paths).

More precisely, we build the graph $G_{r,k}$ as follows. Let $S$ denote the set $\left \{ 1, \dots,\frac{k}{2} \right\}$.
 \begin{itemize}
     \item Consider two cliques of $\frac{k}{2}$ vertices each, $A=\{a_i \mid i\in S\}$ and $B=\{b_i \mid i\in S\}$;
     \item Form the interval graph $IG_{r,k}$ based on $A\cup B$ as in the proof of Proposition~\ref{prop:IG-construction};
     \item To each vertex of $IG_{r,k}$ corresponding to an interval $I_{i,j,s}$ of $IG_{r,k}$, attach a pending path of length~$r$; 
     \item For each nonempty subset $X$ of vertices of each of the cliques $A$ and $B$, add a new vertex $v_X$ that is adjacent only to the vertices in $X$;
     \item To each vertex $v_X$, add a pending path of length $r$.
\end{itemize}

The graph $G_{r,k}$ is chordal, as it can be obtained from an interval graph by iteratively adding simplicial vertices. The vertices from $IG_{r,k}$ and the pending paths attached to them contribute to $\Omega(r^2k^2)$ distinct $r$-profiles. The vertices $v_X$ and their pendant paths contribute to $\Omega(r2^{k/2})$ distinct $r$-profiles.
\end{proof}

\subsection{Graphs of bounded treelength}
\label{sec:tl}

We start with the promised variant of Lemma~\ref{lem:1cliguard}.

\begin{lemma}\label{lem:2cliguard}
Let $G$ be a graph, $r, \ell \in \N$, $A\subseteq V(G)$.
Let $X \subseteq V(G)\setminus A$ and suppose there are two (possibly equal) subsets $S,S'\subseteq V(G)$ such that:
\begin{enumerate}
    \item every path from $X$ to $G-X$ intersects $S\cup S'$; and
    \item the distance in $G$ between any two vertices of $S$ (resp. two vertices of $S'$) is at most~$\ell$,
\end{enumerate}
then the number of $r$-profiles with respect to $A$ of vertices from $X$ is at most $(r+2)^2(\ell+1)^{|A|}$.
\end{lemma}

\begin{proof}
For every vertex $v \in V(G)$, we denote by $d_v$ (resp. $d_v'$) the minimum distance in $G$ from $v$ to some vertex of $S$ (resp.\ $S'$).
Let $v\in X$ and let us count the number of possible $r$-profiles with respect to $A$ it may have.

Let $a\in A$. By assumption every shortest path from $v$ to $a$ intersects one of $S$ and $S'$.
Without loss of generality assume that such a path intersects $S$. Then we have $\dist(v, a) \geq d_v + d_a$. On the other hand, observe that $\dist(v,a) \leq d_v + d_a + \ell$.
If $d_v>r$ there is no choice for the entry of the $r$-profile of $v$ corresponding to $a$: $\cp_r(\dist(v,a)) = +\infty$.

Thus
\[
\cp_r(d_v + d_a) \leq \cp_r \dist(v,a) \leq \cp_r(d_v + d_a  + \ell).\]
The same inequality (with $d'_a$ and $d'_v$ instead of $d_a$ and $d_v$, respectively) would hold if a shortest path from $v$ to $a$ met $S'$.
The distances $d_v$ and $d_v'$ can each take a value in $\intv{0}{r}$ or be larger than $r$, which makes $(r+2)^2$ possible choices. Suppose that the values $d_v$ and $d_v'$ are fixed, then for every $a\in A$, by the inequalities above, $\cp_r \dist(v,a)$ can take at most $\ell+1$ different values. So, in total, the number of $r$-profiles with respect to $A$ of the vertices of $X$ is at most $(r+2)^2(\ell+1)^{|A|}$.
\end{proof}

The proof of Lemma~\ref{lem:1cliguard} can be obtained by following the very same steps as in the proof of Lemma~\ref{lem:2cliguard} for $S'=S$ and $\ell=1$, with the following difference: as in this case $d_v = d'_v$, there are $r+2$ possible choices for this value and not $(r+2)^2$, which yields the bound $(r+2)2^{|A|}$.

We are now ready to prove the main result of this section.

\thlength*

\begin{proof}
Let $\mathcal{T} = (T, \{T_v\}_{v \in V(G)})$ be a tree representation of length at most $\ell$ of $G$.
We start as in the proof of Lemma~\ref{lem:guardingtw}. That is, we root $T$ at some node $s$ and for every vertex $x\in V(G)$ we define $s_x$ as the node of $T_x$ the closest to the root. Let $B = \{s_a \mid
 a \in A\}$ and let $B'$ be the LCA-closure of $B$ in $T$ with root $s$. By Lemma~\ref{lem:lca}, $|B'| \leq 2|B| \leq 2|A|$ and every component $C$ of $T-B'$ has at most two neighbours in $B'$.
 So we can deal separately with the following types of vertices:
 \begin{enumerate}
     \item \emph{Vertices in $\beta(t)$ for some $t \in B'$.} Consider a vertex $v \in \beta(t)$. By the condition on the length of $\mathcal{T}$, for every $a\in A$ we have
     \[
        \dist_G(a, \beta(t)) \leq \dist_G(v,a) \leq \dist_G(a, \beta(t)) + \ell.
     \]
    So in total the vertices in $\beta(t)$ may have at most $(\ell+1)^{|A|}$ different $r$-profiles with respect to~$A$.
     \item \emph{Vertices in the bags of the components of $T-B'$ with the same unique neighbour $t$ in~$B'$.} Let $C$ denote the set of nodes of the components of $T-B'$ that have $t$ as unique neighbour in $B'$.
    Let $X = \bigcup_{c \in C} \beta(c)$.
    We now apply Lemma~\ref{lem:1cliguard} to $G, r, \ell,A,X$ and with $S = \beta(t)$. As $\mathcal{T}$ has length at most $\ell$ and by definition of $B'$ and $X$, the conditions of the lemma are indeed satisfied.
    So the number of $r$-profiles to $A$ of the vertices in $X$ is at most $(r+1)(\ell+1)^{|A|}$.
     \item \emph{Vertices in the bags of a component of $T-B'$ that have two neighbours $t$ and $t'$ in $B'$.} Let $C$ denote the set of nodes of this component and $X = \bigcup_{c \in C} \beta(c)$.
    We apply Lemma~\ref{lem:2cliguard} to $G, r, \ell,A,X$ and with $S = \beta(t)$ and $S' = \beta(t')$ and obtain that the number of $r$-profiles with respect to $A$ of vertices of $X$ is at most $(r+2)^2 (\ell + 1)^{|A|}$.
 \end{enumerate}

 As noted above, $|B'| \leq 2|A|$ so there are at most $2|A|$ sets of vertices to consider in each of the following three cases detailed above. Hence, the number of $r$-profiles with respect to $A$ sums to $O\left (|A|\cdot  \left(r^2(\ell+1)^{|A|} \right ) \right)$, as claimed.
 \end{proof}

 Recall that we gave in Theorem~\ref{th:chordal} a $O(rk2^k+r^2k^3)$ bound on the profile complexity of chordal graphs, i.e. graphs of treelength~1. In this bound, the exponential contribution of $k$ is factor to a term linear in $r$. Below we show that, interestingly, this cannot be achieved for the profile complexity of graphs of bounded treelength. The construction is similar to the one of Proposition~\ref{prop:chordal-construction}.
 
 \begin{proposition}
 Given two integers $r$ and $k$, there exists a graph $H_{r,k}$ of treelength~2 such that
 \[
 \pc_r(H_{r,k},k) \in \Omega(r^22^k).
 \]
 \end{proposition}
 
 \begin{proof}
 The graph $H_{r,k}$ is obtained as follows (we again assume $k$ to be even for convenience). Let $S$ denote the set $\left \{ 1, \dots,\frac{k}{2} \right\}$:
 \begin{itemize}
     \item consider a collection of $\left (2^{\frac{k}{2}}-1\right)^2$ disjoint paths of length $r-1$, each associated to a pair $(X,Y)$ of non-empty subsets of $S$. For any non-empty subsets $X$ and $Y$ of $S$ and any integer $i$ between 1 and $r-1$, let $v_{X,Y,i}$ denote the vertex on position $i$ in the path associated to the pair $(X,Y)$;
     \item add two cliques of $\frac{k}{2}$ vertices each, $\{a_x \mid x\in S\}$ and $\{b_y \mid y\in S\}$;
     \item for any non-empty subsets $X$ and $Z$ of $S$, add an edge between $v_{X,Z,1}$ and those vertices $a_x$ such that $x$ is in $X$;
     \item similarly for any non-empty subsets $Y$ and $Z$ of $S$, add an edge between $v_{Z,Y,r-1}$ and those vertices $b_y$ such that $y$ is in $Y$;
     \item this graph admits a natural path decomposition. Make it of treelength~2 by adding an apex vertex in each bag. Formally, add vertices $u_i$ for $i$ between 0 and $r-1$ such that every vertex $v_{X,Y,i}$ is adjacent to $u_{i-1}$ and $u_i$, $u_0$ is also adjacent to all vertices $a_x$ while $u_{r-1}$ is adjacent to all vertices $b_y$;
     \item for each vertex $v_{X,Y,i}$, add a pending path of length~$r$ $(v_{X,Y,i},p^{X,Y,i}_1,p^{X,Y,i}_2,\ldots,p^{X,Y,i}_r)$.
\end{itemize}  

This graph $H_{r,k}$ has treelength~2. Consider the set $A$ of order $k$ as the union of $\{a_x \mid x\in S\}$ and $\{b_x \mid x\in S\}$. Observe that each vertex $v_{X,Y,i}$ is at distance $i$ from $a_x$ for $x$ in $X$ and $i+1$ from $a_x$ for $x$ not in $X$, and at distance $r-i$ from $b_y$ for $y$ in $Y$ and $r-i+1$ from $b_y$ for $y$ not in $Y$. This gives already $(r-1)(2^{\frac{k}{2}}-1)^2$ different $r$-profiles to $A$. Finally, the pending path on vertex $v_{X,Y,i}$ adds $\frac{r}{2} + |\frac{r}{2} - i|$ new unique $r$-profiles to $A$ for each given triple $(i,X,Y)$. As a whole, we get at least $\frac{r}{2}(r-1)(2^{\frac{k}{2}}-1)^2$ distinct $r$-profiles to $A$. Thus  $\pc_r(H_{r,k},k)\in \Omega(r^22^k)$.
\end{proof}

\section{Open problems}\label{sec:discussion}

We conclude the paper with research directions on the topic of profile and neighbourhood complexities.

\subsection*{Graphs excluding a fixed minor}

While Theorem~\ref{thm:profilesminors} improves on the upper bound for the profile (and neighbourhood) complexity of graphs excluding a fixed complete minor, the following remains open.

\begin{problem}\label{prob:min}
Up to a constant factor, what are the profile and neighbourhood complexity of graphs excluding $K_h$ as a minor?
\end{problem}

More generally, it is an interesting problem to investigate the profile and neighbourhood complexity of graphs excluding $K_s$ as a subdivision. In particular, we believe the bound given by Theorem~\ref{thm:subdiv} is far from optimal.

\begin{problem}
Up to a constant factor, what are the profile and neighbourhood complexities of graphs excluding $K_s$ as a subdivision?
\end{problem}

This problem becomes interesting for $s\geq 5$ (otherwise it is the same as Problem~\ref{prob:min}) and we note that in this case the optimal bound needs to be exponential in~$r$. Indeed, consider the following modification of the construction from~\cite[Theorem~38]{JR}, that defines a subcubic graph $G_k$. Let $A=\{a_1,\ldots,a_k\}$ be a set of vertices. For each vertex $a_i$ in $A$, create a binary tree $T_i$ of height $k$ whose root is $a_i$. For each subset $X$ of $A$, create a binary tree $T_X$ of height $\lceil \log_2 k \rceil$, select $|X|$ of its leaves, and make each of them adjacent to exactly one distinct leaf of one of the trees in $\{T_i, a_i\in X\}$. Now, using the roots of binary trees of the form $T_X$, we can see that $\nc_r(G_k,A)=2^k\geq \frac{2^{r-2}}{r^2}k$, with $r=k+\lceil \log_2 k \rceil + 1$. Since $G_k$ is subcubic, it excludes $K_{5}$ as a subdivision (actually it even excludes $K_{1,4}$ as a subdivision). This construction also shows that, unlike graphs excluding a fixed minor, the distance VC-dimension of graphs excluding a fixed subdivision is unbounded.

\subsection*{Graphs with bounded treewidth}

Theorem~\ref{thm:TW} gives an asymptotically sharp bound for the $r$-profiles of graphs with bounded treewidth. What about the neighbourhood complexity? Recall that Theorem~\ref{thm:TW} implies that for every integers $r,k\ge 0$, every graph $G$ with treewidth at most $t$ we have $\nc_r(G,k)\in O_t(r^{t+1}k).$ On the other hand, by Corollary~37 of~\cite{JR}, there is a graph $G$ of treewidth at most $t$ and $A\subseteq V(G)$ such that $\nc_r(G,A) \in \Omega_t(r^{t}|A|).$ It thus remains to close the gap.

\begin{problem}
Up to a constant factor, what is the neighbourhood complexity of graphs with treewidth at most $t$?
\end{problem}
\noindent

\subsection*{Graphs with bounded simple treewidth}

A graph is a \emph{$k$-tree} if it is either a clique of order $k + 1$ or can be
obtained from a smaller $k$-tree by adding a vertex and making it adjacent to $k$ pairwise-adjacent
vertices. A \emph{simple $k$-tree} is a $k$-tree built with the restriction that when
adding a new vertex, the clique to which we make it adjacent cannot have been used when
adding some other vertex. The \emph{simple
treewidth}, $\stw(G)$, of a graph $G$ is the smallest $k$ such that $G$ is a subgraph of a simple $k$-tree. Since the treewidth of a graph $G$ can be equivalently defined as the smallest $k$ such that $G$ is a subgraph of a $k$-tree, we have $\tw(G)\le \stw(G)$ for every graph $G$. Graphs with $\stw(G)\le 1$ are the disjoint union of paths, and graphs with $\stw(G)\le 2$ form exactly the class of outerplanar graphs. Recall that we obtained an $O(r^2k)$ bound on the profile complexity of outerplanar graphs in Theorem~\ref{th:outer}. We conjecture that this result can be generalized to graphs of bounded simple treewidth as follows.

\begin{restatable}{conjecture}{conjsimpl}\label{conj:simpleTW}
Let $t,r$ be two positive integers and let $G$ be a graph of simple treewidth at most $t$. Then, $\pc_r(G,k)\in O_t(r^{t}k)$.
\end{restatable}

This conjecture holds for $t=2$ by Proposition~\ref{th:outer}, and holds trivially for $t=1$: since a graph $G$ with $\stw(G)= 1$ is the disjoint union of paths, for every vertex $v$ in $G$ we have $|N_r[v]|\le 2r+1$. This immediately gives $\pc_r(G,k)\le (2r+1)k$. 
Since $\tw(G)\le \stw(G)$ for every graph $G$, the construction  of Corollary~37 of \cite{JR} attests that, if true, this conjectured bound would be asymptotically tight.

\section*{Acknowledgements}

We thank the two anonymous referees for their constructive remarks, which significantly improved the quality and clarity of the manuscript.

The project has received funding from the following grants: FONDECYT/ANID Iniciaci\'on en Investigaci\'on Grant 11201251; Programa Regional MATH-AMSUD MATH230035; IDEX-ISITE initiative CAP 20-25 (ANR-16-IDEX-0001); International Research Center ``Innovation Transportation and Production Systems'' of the I-SITE CAP 20-25; ANR project GRALMECO (ANR-21-CE48-0004); European Union (ERC, POCOCOP, 101071674). Views and opinions expressed are however those of the authors only and do not necessarily reflect those of the European Union or the European Research Council Executive Agency. Neither the European Union nor the granting authority can be held responsible for them.

\bibliographystyle{alpha}
\bibliography{references}

\newcommand{\etalchar}[1]{$^{#1}$}
\begin{thebibliography}{HOdMQ{\etalchar{+}}17}

\bibitem[AMS19]{AMS19}
Noga Alon, Guy Moshkovitz, and Noam Solomon.
\newblock Traces of hypergraphs.
\newblock {\em J. Lond. Math. Soc.}, 100(2):498--517, 2019.

\bibitem[BBB{\etalchar{+}}21]{DBLP:journals/jacm/BonamyBBCGKRST21}
Marthe Bonamy, {\'{E}}douard Bonnet, Nicolas Bousquet, Pierre Charbit, Panos
  Giannopoulos, Eun~Jung Kim, Pawe\l{} Rz\k{a}\.{z}ewski, Florian Sikora, and
  St{\'{e}}phan Thomass{\'{e}}.
\newblock {EPTAS} and subexponential algorithm for maximum clique on disk and
  unit ball graphs.
\newblock {\em J. {ACM}}, 68(2):9:1--9:38, 2021.

\bibitem[BBGR24]{berthe2024subexponential}
Ga{\'e}tan Berthe, Marin Bougeret, Daniel Gon{\c{c}}alves, and Jean-Florent
  Raymond.
\newblock Subexponential algorithms in geometric graphs via the subquadratic
  grid minor property: The role of local radius.
\newblock In {\em 19th Scandinavian Symposium and Workshops on Algorithm Theory
  (SWAT 2024)}. Schloss Dagstuhl--Leibniz-Zentrum f{\"u}r Informatik, 2024.

\bibitem[BDF{\etalchar{+}}18]{LFetal}
Laurent Beaudou, Peter Dankelmann, Florent Foucaud, Michael~A. Henning, Arnaud
  Mary, and Aline Parreau.
\newblock Bounding the order of a graph using its diameter and metric
  dimension: {A} study through tree decompositions and {VC} dimension.
\newblock {\em {SIAM} J. Discret. Math.}, 32(2):902--918, 2018.

\bibitem[BEHW89]{DBLP:journals/jacm/BlumerEHW89}
Anselm Blumer, Andrzej Ehrenfeucht, David Haussler, and Manfred~K. Warmuth.
\newblock Learnability and the {V}apnik-{C}hervonenkis dimension.
\newblock {\em J. {ACM}}, 36(4):929--965, 1989.

\bibitem[BFLP24]{DBLP:journals/ejc/BonnetFLP24}
{\'{E}}douard Bonnet, Florent Foucaud, Tuomo Lehtil{\"{a}}, and Aline Parreau.
\newblock Neighbourhood complexity of graphs of bounded twin-width.
\newblock {\em Eur. J. Comb.}, 115:103772, 2024.

\bibitem[BFLP26]{berthe2026neighbourhoodcomplexityidentificationproblems}
Gaétan Berthe, Florent Foucaud, Tuomo Lehtilä, and Aline Parreau.
\newblock Neighbourhood complexity and identification problems for graphs of
  bounded treewidth and pathwidth.
\newblock {\em arXiv e-print}, 2026.
\newblock \url{https://arxiv.org/abs/2607.16889}.

\bibitem[BG25]{BG25}
Marthe Bonamy and Colin Geniet.
\newblock $\chi$-boundedness and neighbourhood complexity of bounded
  merge-width graphs.
\newblock {\em arXiv e-print}, 2025.
\newblock \url{https://arxiv.org/abs/2504.08266}.

\bibitem[BLL{\etalchar{+}}15]{BLLPT15}
Nicolas Bousquet, Aur{\'{e}}lie Lagoutte, Zhentao Li, Aline Parreau, and
  St{\'{e}}phan Thomass{\'{e}}.
\newblock Identifying codes in hereditary classes of graphs and {VC}-dimension.
\newblock {\em {SIAM} J. Discret. Math.}, 29(4):2047--2064, 2015.

\bibitem[BS24]{DBLP:journals/combinatorica/BergerS24}
Eli Berger and Paul~D. Seymour.
\newblock Bounded-diameter tree-decompositions.
\newblock {\em Comb.}, 44(3):659--674, 2024.

\bibitem[BT15]{DBLP:journals/dm/BousquetT15}
Nicolas Bousquet and St{\'{e}}phan Thomass{\'{e}}.
\newblock {VC}-dimension and {E}rd{\H{o}}s-{P}{\'{o}}sa property.
\newblock {\em Discret. Math.}, 338(12):2302--2317, 2015.

\bibitem[CEV07]{DBLP:journals/dcg/ChepoiEV07}
Victor Chepoi, Bertrand Estellon, and Yann Vax{\`{e}}s.
\newblock Covering planar graphs with a fixed number of balls.
\newblock {\em Discret. Comput. Geom.}, 37(2):237--244, 2007.

\bibitem[CFPW24]{DBLP:journals/fuin/ChakrabortyFPW24}
Dipayan Chakraborty, Florent Foucaud, Aline Parreau, and Annegret Wagler.
\newblock On three domination-based identification problems in block graphs.
\newblock {\em Fundam. Informaticae}, 191(3-4):197--229, 2024.

\bibitem[CKM{\etalchar{+}}25]{subchromatic}
Pedro~P. Cortés, Pankaj Kumar, Benjamin Moore, Patrice~Ossona de~Mendez, and
  Daniel~A. Quiroz.
\newblock Subchromatic numbers of powers of graphs with excluded minors.
\newblock {\em Discret. Math.}, 348(114377), 2025.

\bibitem[DEM{\etalchar{+}}24]{DBLP:conf/focs/DreierEMMPT24}
Jan Dreier, Ioannis Eleftheriadis, Nikolas M{\"{a}}hlmann, Rose McCarty, Michal
  Pilipczuk, and Szymon Torunczyk.
\newblock First-order model checking on monadically stable graph classes.
\newblock In {\em 65th {IEEE} Annual Symposium on Foundations of Computer
  Science, {FOCS} 2024, Chicago, IL, USA, October 27-30, 2024}, pages 21--30.
  {IEEE}, 2024.

\bibitem[DG07]{DBLP:journals/dm/DourisboureG07}
Yon Dourisboure and Cyril Gavoille.
\newblock Tree-decompositions with bags of small diameter.
\newblock {\em Discret. Math.}, 307(16):2008--2029, 2007.

\bibitem[DGP26]{DBLP:journals/corr/abs-2601-14906}
Jan Dreier, Jakub Gajarsk{\`y}, and Micha{\l} Pilipczuk.
\newblock Efficient reversal of transductions of sparse graph classes.
\newblock In {\em Proceedings of the 58th Annual ACM Symposium on Theory of
  Computing (STOC 2026)}, pages 1180--1191, 2026.

\bibitem[DHV22]{DBLP:journals/siamcomp/DucoffeHV22}
Guillaume Ducoffe, Michel Habib, and Laurent Viennot.
\newblock Diameter, eccentricities and distance oracle computations on
  \emph{H}-minor free graphs and graphs of bounded (distance)
  vapnik-chervonenkis dimension.
\newblock {\em {SIAM} J. Comput.}, 51(5):1506--1534, 2022.

\bibitem[DMM{\etalchar{+}}26]{dreier2026neighborhoodcomplexityradius1mergewidth}
Jan Dreier, Nikolas Mählmann, Rose McCarty, Michał Pilipczuk, and Szymon
  Toruńczyk.
\newblock Neighborhood complexity and radius-1 merge-width in monadically
  dependent graph classes.
\newblock {\em arXiv preprint}, 2026.
\newblock \url{https://arxiv.org/abs/2607.10941}.

\bibitem[DPUY22]{DPUY22}
Zden\v{e}k Dvo\v{r}\'{a}k, Jakub Pek\'{a}rek, Torsten Ueckerdt, and Yelena
  Yuditsky.
\newblock {Weak Coloring Numbers of Intersection Graphs}.
\newblock In Xavier Goaoc and Michael Kerber, editors, {\em 38th International
  Symposium on Computational Geometry (SoCG 2022)}, volume 224 of {\em Leibniz
  International Proceedings in Informatics (LIPIcs)}, pages 39:1--39:15,
  Dagstuhl, Germany, 2022. Schloss Dagstuhl -- Leibniz-Zentrum f{\"u}r
  Informatik.

\bibitem[EGK{\etalchar{+}}17]{Eickmeyer17}
Kord Eickmeyer, Archontia~C. Giannopoulou, Stephan Kreutzer, O{-}joung Kwon,
  Michał Pilipczuk, Roman Rabinovich, and Sebastian Siebertz.
\newblock Neighborhood complexity and kernelization for nowhere dense classes
  of graphs.
\newblock In Ioannis Chatzigiannakis, Piotr Indyk, Fabian Kuhn, and Anca
  Muscholl, editors, {\em 44th International Colloquium on Automata, Languages,
  and Programming, {ICALP} 2017, July 10-14, 2017, Warsaw, Poland}, volume~80
  of {\em LIPIcs}, pages 63:1--63:14. Schloss Dagstuhl - Leibniz-Zentrum
  f{\"{u}}r Informatik, 2017.

\bibitem[FGL{\etalchar{+}}22]{DBLP:journals/dam/FrongilloGLTY22}
Rafael~M. Frongillo, Jesse Geneson, Manuel~E. Lladser, Richard~C. Tillquist,
  and Eunjeong Yi.
\newblock Truncated metric dimension for finite graphs.
\newblock {\em Discret. Appl. Math.}, 320:150--169, 2022.

\bibitem[FGN{\etalchar{+}}13]{DBLP:journals/jgt/FoucaudGNPV13}
Florent Foucaud, Sylvain Gravier, Reza Naserasr, Aline Parreau, and Petru
  Valicov.
\newblock Identifying codes in line graphs.
\newblock {\em J. Graph Theory}, 73(4):425--448, 2013.

\bibitem[FLMS12]{Fominetal}
Fedor~V. Fomin, Daniel Lokshtanov, Neeldhara Misra, and Saket Saurabh.
\newblock Planar {F}-deletion: Approximation, kernelization and optimal {FPT}
  algorithms.
\newblock In {\em 53rd Annual {IEEE} Symposium on Foundations of Computer
  Science, {FOCS} 2012, New Brunswick, NJ, USA, October 20-23, 2012}, pages
  470--479. {IEEE} Computer Society, 2012.

\bibitem[FMN{\etalchar{+}}17]{FMNPV17}
Florent Foucaud, George~B. Mertzios, Reza Naserasr, Aline Parreau, and Petru
  Valicov.
\newblock Identification, location-domination and metric dimension on interval
  and permutation graphs. {I}. bounds.
\newblock {\em Theor. Comput. Sci.}, 668:43--58, 2017.

\bibitem[Fra83]{F83}
Peter Frankl.
\newblock On the trace of finite sets.
\newblock {\em J. Comb. Theory {A}}, 34(1):41--45, 1983.

\bibitem[GHO{\etalchar{+}}17]{GAJARSKY}
Jakub Gajarský, Petr Hliněný, Jan Obdržálek, Sebastian Ordyniak, Felix
  Reidl, Peter Rossmanith, Fernando {Sánchez Villaamil}, and Somnath Sikdar.
\newblock Kernelization using structural parameters on sparse graph classes.
\newblock {\em Journal of Computer and System Sciences}, 84:219--242, 2017.

\bibitem[HK21]{VdHK21}
Jan van~den Heuvel and H.A. Kierstead.
\newblock Uniform orderings for generalized coloring numbers.
\newblock {\em Eur. J. Comb.}, 91:103214, 2021.
\newblock Colorings and structural graph theory in context (a tribute to Xuding
  Zhu).

\bibitem[HM76]{Harary76}
Frank Harary and Robert Melter.
\newblock On the metric dimension of a graph.
\newblock {\em Ars Comb.}, 2:191--195, 1976.

\bibitem[HMP{\etalchar{+}}10]{DBLP:journals/combinatorics/HernandoMPSW10}
Carmen Hernando, Merc{\`{e}} Mora, Ignacio~M. Pelayo, Carlos Seara, and
  David~R. Wood.
\newblock Extremal graph theory for metric dimension and diameter.
\newblock {\em Electron. J. Comb.}, 17(1), 2010.

\bibitem[HOdMQ{\etalchar{+}}17]{gencol}
Jan van~den Heuvel, Patrice Ossona~de Mendez, Daniel Quiroz, Roman Rabinovich,
  and Sebastian Siebertz.
\newblock On the generalised colouring numbers of graphs that exclude a fixed
  minor.
\newblock {\em Eur. J. Comb.}, 66:129--144, 2017.

\bibitem[Hon09]{DBLP:journals/ejc/Honkala09}
Iiro~S. Honkala.
\newblock On r-locating-dominating sets in paths.
\newblock {\em Eur. J. Comb.}, 30(4):1022--1025, 2009.

\bibitem[Jea26]{lobstein2012watching}
Devin Jean.
\newblock Watching systems, identifying, locating-dominating and discriminating
  codes in graphs: a bibliography.
\newblock {\em Published electronically at
  \url{https://dragazo.github.io/bibdom/main.pdf}}, 2026.

\bibitem[JR24]{JR}
Gwena{\"{e}}l Joret and Cl{\'{e}}ment Rambaud.
\newblock Neighborhood complexity of planar graphs.
\newblock {\em Comb.}, 44(5):1115--1148, 2024.

\bibitem[KKR{\etalchar{+}}97]{VCdim-graph}
Evangelos Kranakis, Danny Krizanc, Berthold Ruf, Jorge Urrutia, and Gerhard~J.
  Woeginger.
\newblock The {VC}-dimension of set systems defined by graphs.
\newblock {\em Discret. Appl. Math.}, 77(3):237--257, 1997.

\bibitem[KPPS25]{DBLP:conf/icalp/KlukPPS25}
Kacper Kluk, Marcin Pilipczuk, Michal Pilipczuk, and Giannos Stamoulis.
\newblock Faster diameter computation in graphs of bounded euler genus.
\newblock In Keren Censor{-}Hillel, Fabrizio Grandoni, Jo{\"{e}}l Ouaknine, and
  Gabriele Puppis, editors, {\em 52nd International Colloquium on Automata,
  Languages, and Programming, {ICALP} 2025, Aarhus, Denmark, July 8-11, 2025},
  volume 334 of {\em LIPIcs}, pages 109:1--109:19. Schloss Dagstuhl -
  Leibniz-Zentrum f{\"{u}}r Informatik, 2025.

\bibitem[KPRS16]{KPRS16}
Stephan Kreutzer, Michal Pilipczuk, Roman Rabinovich, and Sebastian Siebertz.
\newblock {The Generalised Colouring Numbers on Classes of Bounded Expansion}.
\newblock In Piotr Faliszewski, Anca Muscholl, and Rolf Niedermeier, editors,
  {\em 41st International Symposium on Mathematical Foundations of Computer
  Science (MFCS 2016)}, volume~58 of {\em Leibniz International Proceedings in
  Informatics (LIPIcs)}, pages 85:1--85:13, Dagstuhl, Germany, 2016. Schloss
  Dagstuhl -- Leibniz-Zentrum f{\"u}r Informatik.

\bibitem[KY03]{KY}
Hal~A. Kierstead and Daqing Yang.
\newblock Orderings on graphs and game coloring number.
\newblock {\em Order}, 20(3):255--264, 2003.

\bibitem[KY21]{KuziakYeroSurvey}
Dorota Kuziak and Ismael~G. Yero.
\newblock Metric dimension related parameters in graphs: A survey on
  combinatorial, computational and applied results.
\newblock {\em arXiv preprint arXiv:2107.04877}, 2021.

\bibitem[LHC20]{lobstein2020locating}
Antoine Lobstein, Olivier Hudry, and Ir{\`e}ne Charon.
\newblock Locating-domination and identification.
\newblock {\em Topics in Domination in Graphs}, pages 251--299, 2020.

\bibitem[LPS{\etalchar{+}}22]{lokshtanov2022subexponential}
Daniel Lokshtanov, Fahad Panolan, Saket Saurabh, Jie Xue, and Meirav Zehavi.
\newblock Subexponential parameterized algorithms on disk graphs (extended
  abstract).
\newblock In {\em Proceedings of the 2022 Annual ACM-SIAM Symposium on Discrete
  Algorithms (SODA)}, pages 2005--2031. SIAM, 2022.

\bibitem[MV17]{Mustafa2017}
Nabil~H. Mustafa and Kasturi Varadarajan.
\newblock Epsilon-approximations and epsilon-nets.
\newblock In J.~E. Goodman, J.~O'Rourke, and C.~D. Tóth, editors, {\em
  Handbook of Discrete and Computational Geometry}. CRC Press LLC, 2017.

\bibitem[NdM12]{sparsity}
Jaroslav Ne{\v{s}}et{\v{r}}il and Patrice~Ossona de~Mendez.
\newblock {\em Sparsity - Graphs, Structures, and Algorithms}, volume~28 of
  {\em Algorithms and combinatorics}.
\newblock Springer, 2012.

\bibitem[PP20]{Paszke20}
Adam Paszke and Michał Pilipczuk.
\newblock {VC} density of set systems definable in tree-like graphs.
\newblock In Javier Esparza and Daniel Kr{\'{a}}l', editors, {\em 45th
  International Symposium on Mathematical Foundations of Computer Science,
  {MFCS} 2020, August 24-28, 2020, Prague, Czech Republic}, volume 170 of {\em
  LIPIcs}, pages 78:1--78:13. Schloss Dagstuhl - Leibniz-Zentrum f{\"{u}}r
  Informatik, 2020.

\bibitem[RS84]{RS84}
Douglas Rall and Peter~J. Slater.
\newblock On location-domination numbers for certain classes of graphs.
\newblock {\em Congr. Numer.}, 45:97--106, 1984.

\bibitem[RVS19]{Reidl19}
Felix Reidl, Fernando~S{\'{a}}nchez Villaamil, and Konstantinos~S.
  Stavropoulos.
\newblock Characterising bounded expansion by neighbourhood complexity.
\newblock {\em Eur. J. Comb.}, 75:152--168, 2019.

\bibitem[Sau72]{Sauer72}
Norbert Sauer.
\newblock On the density of families of sets.
\newblock {\em J. Comb. Theory {A}}, 13(1):145--147, 1972.

\bibitem[She72]{Shelah72}
Saharon Shelah.
\newblock {A combinatorial problem; stability and order for models and theories
  in infinitary languages.}
\newblock {\em Pacific Journal of Mathematics}, 41(1):247 -- 261, 1972.

\bibitem[Sie26]{Siebertz25}
Sebastian Siebertz.
\newblock On the generalized coloring numbers.
\newblock {\em Computer Science Review}, 59:100855, 2026.

\bibitem[Sla75]{S75}
Peter~J. Slater.
\newblock Leaves of trees.
\newblock In {\em Proceedings of the Sixth Southeastern Conference on
  Combinatorics, Graph Theory, and Computing (Florida Atlantic Univ., Boca
  Raton, Fla., 1975)}, pages 549--559. Congressus Numerantium, No. XIV,
  Winnipeg, Man., 1975. Utilitas Math.

\bibitem[Sla87]{DBLP:journals/networks/Slater87}
Peter~J. Slater.
\newblock Domination and location in acyclic graphs.
\newblock {\em Networks}, 17(1):55--64, 1987.

\bibitem[Sok23]{DBLP:journals/combinatorics/Sokolowski23}
Marek Sokolowski.
\newblock Bounds on half graph orders in powers of sparse graphs.
\newblock {\em Electron. J. Comb.}, 30(2), 2023.

\bibitem[ST06]{Six1999}
Janet~M. Six and Ioannis~G. Tollis.
\newblock A framework and algorithms for circular drawings of graphs.
\newblock {\em Journal of Discrete Algorithms}, 4(1):25--50, 2006.

\bibitem[TFL21]{TillquistSurvey}
Richard~C. Tillquist, Rafael~M. Frongillo, and Manuel~E. Lladser.
\newblock Getting the lay of the land in discrete space: A survey of metric
  dimension and its applications.
\newblock {\em arXiv preprint arXiv:2104.07201}, 2021.

\bibitem[VC71]{VCdim}
Vladimir~N. Vapnik and Alexey~Y. Chervonenkis.
\newblock On the uniform convergence of relative frequencies of events to their
  probabilities.
\newblock {\em Theory of Probability \& Its Applications}, 16(2):264--280,
  1971.

\end{thebibliography}

\end{document}